\pdfoutput=1 
\newif\ifSC
\SCfalse
\ifSC
\documentclass[12pt, draftclsnofoot, onecolumn]{IEEEtran}
\else
\documentclass[10pt, journal,twocolumn]{IEEEtran}
\fi
\usepackage{amsmath}
\usepackage{mathtools}
\usepackage{amsmath}
\usepackage{amssymb}
\usepackage{cite}
\usepackage{array}
\usepackage{amsthm}
\usepackage{amsmath}
\usepackage{tabulary}
\usepackage{multirow}
\usepackage{subfig}
\usepackage[margin=1in]{geometry}
\usepackage{graphicx}
\usepackage[font=small]{caption}
\usepackage{etoolbox}
\usepackage{color}

\IEEEoverridecommandlockouts
\usepackage{pifont}
\usepackage{pgfplots}
\usepackage{bbm}
\usepackage{paralist}
\usepackage{url}
\usepackage{amssymb}
\usepackage[]{algorithm2e}
\usepackage{algpseudocode}
\usepackage{algorithmicx}
\usepackage{dsfont}
\makeatletter
\renewcommand{\@IEEEsectpunct}{.\ \,}% Modified from {:\ \,}
\makeatother
\graphicspath{ {./figs/} }

\newcommand{\expect}[1]{{\mathbb{E}\left[{#1}\right]}}
\newcommand{\cexpect}[2]{{\mathbb{E}_{#2}\left[{#1}\right]}}
\newcommand{\pr}{{\mathbb{P}}}

\newcommand{\x}{{x}}
\newcommand{\us}{u}

\newcommand{\sdpl}{\xi}
\newcommand{\meen}{m}

\newcommand{\ple}{{\alpha}}

\newcommand{\power}[1]{\mathrm{P}_{#1}}
\newcommand{\res}{\mathrm{B}}

\newcommand{\dnsty}{{\lambda_\mathrm{BS}}}

\newcommand{\fsldb}{\beta}

\newcommand{\userdnsty}{\lambda_\mathrm{UE}}

\newcommand{\Q}{\mathsf{Q}} 
 
\newcommand{\NBS}{\mathrm{N}_\mathrm{BS}}
\newcommand{\NMS}{\mathrm{N}_\mathrm{UE}}
\newcommand{\NRF}{\mathrm{U}_\mathrm{M}}
\newcommand{\aMS}{\mathrm{\bf a}_\mathrm{UE}}
\newcommand{\aBS}{\mathrm{\bf a}_\mathrm{BS}}
\newcommand{\AOD}[3]{\theta_{#1,#2,#3}}
\newcommand{\AOA}[3]{\phi_{#1,#2,#3}}
\newcommand{\BS}{\mathrm{BS}}
\newcommand{\UE}{\mathrm{UE}}
\newcommand{\LOS}{\mathrm{LOS}}

\newcommand{\dmax}{\mathrm{D}}
\newcommand{\bbprecoder}{\mathrm{\bf F}^\mathrm{BB}_{\x}}
\newcommand{\rfprecoder}{\mathrm{\bf F}^\mathrm{RF}_\x}
\newcommand{\bbcolumn}[2]{\mathrm{\bf f}^{\mathrm{BB}}_{#2,#1}}
\newcommand{\ubbprecoder}{\mathrm{\bf W}^\mathrm{BB}_{\us}}
\newcommand{\urfprecoder}{\mathrm{\bf W}^\mathrm{RF}_\us}
\newcommand{\ubbcolumn}[2]{\mathrm{\bf w}^{\mathrm{BB}}_{#2,#1}}
\newcommand{\rfcolumn}[2]{\mathrm{\bf f}^{\mathrm{RF}}_{#2,#1}}
\newcommand{\heffective}{\overline{\bf h}^{*}_{\x,\us}}

\newcommand{\SINR}{\mathtt{SINR}}

\newcommand{\SNR}{\mathtt{SNR}}

\newcommand{\PPP}{\Phi_\mathrm{BS}}

\newcommand{\PPPu}{\Phi_\mathrm{UE}}

\newcommand{\plos}{p_{\mathrm{LOS}}}

\newtheorem{thm}{{\bf Theorem}}
\newtheorem{cor}{Corollary}
\newtheorem{rem}{Remark}
\newtheorem{lem}{Lemma}
\newtheorem{prop}{Proposition}
\theoremstyle{definition}
\newtheorem{definition}{Definition}

\newtoggle{SC}
\togglefalse{SC}
\ifSC
\toggletrue{SC}
\fi

\begin{document}

%----------------------------------------------------------------------
% Title Information, Abstract and Keywords
%----------------------------------------------------------------------
\title{A Comparison of MIMO Techniques in Downlink Millimeter Wave Cellular Networks with Hybrid Beamforming}
\author{Mandar N. Kulkarni, Amitava Ghosh and Jeffrey G. Andrews \thanks{This work was supported by Nokia. Part of this work appeared in Proc. {\em IEEE Asilomar}, Nov. 2015\cite{KulAlk15}. M. N. Kulkarni and J. G. Andrews are with the University of Texas at Austin, TX. A. Ghosh is with Nokia Networks, Arlington Heights, IL. Email: {\tt mandar.kulkarni@utexas.edu}, {\tt amitava.ghosh@nokia.com},{\tt jandrews@ece.utexas.edu}}}
\maketitle

\begin{abstract}
Large antenna arrays will be needed in future millimeter wave (mmWave) cellular networks, enabling a large number of different possible antenna architectures and multiple-input multiple-output (MIMO) techniques. It is still unclear which MIMO technique is most desirable as a function of different network parameters. This paper, therefore, compares the coverage and rate performance of hybrid beamforming enabled multi-user (MU) MIMO and single-user spatial multiplexing (SM) with single-user analog beamforming (SU-BF). A stochastic geometry model for coverage and rate analysis is proposed for MU-MIMO mmWave cellular networks, taking into account important mmWave-specific hardware constraints for hybrid analog/digital precoders and combiners, and a blockage-dependent channel model which is sparse in angular domain. The analytical results highlight the coverage, rate and power consumption tradeoffs in multiuser mmWave networks. With perfect channel state information at the transmitter and round robin scheduling, MU-MIMO is usually a better choice than SM or SU-BF in mmWave cellular networks. This observation, however, neglects any overhead due to channel acquisition or computational complexity. Incorporating the impact of such overheads, our results can be re-interpreted so as to quantify the minimum allowable efficiency of MU-MIMO to provide higher rates than SM or SU-BF. 
\end{abstract}

\section{Introduction}
A classical question in multi-antenna wireless communications has been to determine which MIMO technique performs better in different scenarios, for example based on the channel and interference characteristics. At mmWave frequencies, several important new factors must be considered, due to different hardware constraints on the precoders/combiners and a significantly different outdoor channel, which is both blockage-dependent and sparse (low rank) \cite{Akd14,Sun14,Aya14}. In order to compensate for the large near-field path loss, SU-BF has been the primary focus of several existing system capacity evaluations for mmWave cellular networks \cite{roh14,Akd14,BaiHea14,Ghosh14}. However, recently there has been significant work on enabling MU-MIMO and SM under different antenna architectures that respect the necessary hardware constraints at mmWave frequencies. Hybrid analog-digital  precoders and combiners, receivers with low resolution analog to digital converters, and continuous aperture phase MIMO with lens-based beamformers (also called CAP-MIMO) are prominent antenna architectures being considered\cite{AlkMo14,Mo15,Say10}. Most existing studies on mmWave MIMO, except for SU-BF, rely on single cell analysis for evaluating performance of the MIMO techniques and/or system level simulations for understanding the impact of base station (BS) deployment scenarios or blockages in the environment on the coverage and rate performance. Although analytical models for studying coverage and rate in SU-BF mmWave networks have been studied\cite{BaiHea14,SinJSAC14}, these cannot be directly used for studying other MIMO techniques like MU-MIMO and SM, as will be explained in Section~\ref{sec:background}.

The goals of this paper are two-fold. First, we propose a stochastic geometry-based model to study coverage and per user rate distribution in fully-connected hybrid beamforming-enabled MU-MIMO mmWave cellular networks. Second, we use this analytical model as a tool for comparing coverage, rate and power consumption for MU-MIMO, SM and SU-BF mmWave cellular networks. 

\subsection{Background and Related Work}
\label{sec:background}
Conventionally, BSs are equipped with fully-digital baseband processing. However, this approach requires a radio frequency (RF) chain per antenna which is impractical for mmWave BSs equipped with large antenna arrays. Fully analog solutions, on the other hand, require only a single RF chain for the entire antenna array but have no capability of digital processing. Hybrid beamforming strikes a balance between these two solutions, wherein the number of RF chains can be designed to be between 1 (analog beamforming) and the number of antennas (digital beamforming). In a fully-connected architecture, each RF chain has phase shifters connected to all antennas in the array. On the other hand, in the array of sub-arrays architecture, the entire array is divided into sub-arrays and all antennas in a sub-array are connected via phase shifters to exactly one RF chain. The fully-connected architecture has higher beamforming gain than array of sub-arrays, for a fixed number of antennas. However, the power consumption and hardware complexity of precoder/combiner is lower in the latter. 
%The tradeoff in power efficiency and data rates in these two architectures has not been formally studied, although the qualitative discussion in \cite{Han15} suggests that for $\mathrm{N}$ antennas at BS, the spectral efficiency with the array of sub-arrays architecture is roughly $\mathrm{N}\log\mathrm{N}$ bits/second/Hz lower than with the fully-connected approach at the expense of increased complexity and power consumption. For large antenna arrays, this gap in spectral efficiency is significant. Further, 
With low-complexity yet near optimal precoding/combining algorithms for MU-MIMO and SM being proposed with the fully-connected architecture\cite{Aya14,Alk14}, this approach looks promising for practical implementation and is the focus of our discussion. 

In \cite{Alk14}, a joint baseband-RF precoder solution for MU-MIMO was proposed and proven to be asymptotically optimal as the number of antennas become large. Using this scheme, it was observed that MU-MIMO can offer higher sum rates than SU-BF. Another simulation-based work \cite{Voo14} highlighted that per user rates, including the cell edge rates, can be much higher with MU-MIMO with appropriate user pairing. It was observed that exploiting polarization diversity for two stream transmission to each user further enhances the gains in using MU-MIMO. This is one particular way in which SM gains can be obtained in tandem with MU-MIMO. 
%Note that for exploiting polarization diversity one would need two different sub-arrays employing vertical and horizontal polarization at the transmitter and the receiver. Note that this is a special case of array of subarrays architecture. 
Another way to get SM gains would be to rely on the scatterers in the environment\cite{Aya14,Sun14}.  The simulations in \cite{Aya14,Sun14} showed that SM and SU-BF could work in tandem to improve capacity. However, all these works implicitly neglected the aspect of power consumption at BSs and UEs when comparing the MIMO techniques. If we were to compare the coverage and rate performance of SU-BF and MU-MIMO or SM with fixed power consumption per unit area and fixed number of antennas per BS, we can deploy a much denser mmWave network with SU-BF than MU-MIMO or SM. This significantly affects the comparisons as will be shown in Section~\ref{sec:results}, since unlike in conventional cellular networks\cite{andganbac11}, densifying a mmWave network boosts the coverage and capacity notably\cite{Ghosh14,BaiHea14}. 

The above mentioned studies either rely on system-level simulations or on single cell analysis. There is no analytical model for MU-MIMO or SM mmWave networks that incorporates the impact of hybrid precoders and combiners and the channel sparsity. 
%It has been widely accepted now, that stochastic geometry models are state of the art in analysis of wireless networks taking into account spatial distribution of devices in the network \cite{andganbac11,BacBook09}. Thus, we use stochastic geometry as a tool for coverage and rate analysis of mmWave cellular networks employing MU-MIMO and SU-BF. Note that the approach in this paper can also be used for modeling coverage in hybrid beamforming enabled SM mmWave cellular networks and this is deferred as a scope for future work. The proposed analytical model is then used as a tool for comparing MU-MIMO, SU-BF and SM. 
Analysis for MIMO cellular networks has conventionally been done by capturing the impact of linear precoding and combining into the distribution of an effective small scale fading random variable\cite{Hea11,DhiKountAnd13, Zhu14, Gupta14}.  In \cite{Hea11,DhiKountAnd13}, it was shown that Gamma distribution can be used to model the small scale fading gain in MU-MIMO cellular networks employing ZF precoding. Most successive analytical studies on MU-MIMO cellular networks using stochastic geometry have relied on this result, for example \cite{Zhu14, Gupta14}. However, justifying this result assumes fully digital processing and full rank Rayleigh fading channels. At mmWave frequencies the channel is expected to be sparse and blockage-dependent\cite{Akd14,Sun14,SamRap14,Sam16}. Thus, the full rank assumption is far from reality. A recent work \cite{Bai2015} proposes an analytical model for $\SINR$ (signal to interference plus noise ratio) in MU-MIMO mmWave cellular networks but assumes fully digital processing. But, as described earlier, fully digital processing is also not realistic at mmWave. Analysis of multiuser mmWave cellular networks, thus, demands a new approach. Also, other existing analytical models for SU-BF enabled mmWave networks assume an equivalent SISO-like system with directional antenna gains by abstracting underlying signal level details\cite{BaiHea14, SinJSAC14}. Further, the analysis in these papers is done for single path channels. An analytical framework that can be used as a tool for comparing with different MIMO techniques needs to incorporate multipath in the channel, which is a primary feature enabling SM. The key contributions in this work are as follows. 

\subsection{Contributions}
\subsubsection{Tractable Model for Coverage and Rate in MU-MIMO mmWave Cellular Networks}
The analytical model captures the following mmWave-specific features: (i) precoding and combining with hybrid beamforming, and (ii) sparse blockage-dependent multipath channel model. For simplicity the channel model is assumed to be non-selective in both time and frequency to focus only on the spatial aspects. Using Monte-Carlo simulations, the model is shown to be reasonably accurate for a large number of antennas at the BSs and user equipments (UEs) in noise-limited scenarios. In interference-limited scenarios, upper and lower bounds to the distribution of the proposed approximate $\SINR$ model are derived under some assumptions and validated with Monte-Carlo simulations. The fact that our model incorporates different channel rank for line-of-sight (LOS) and non-LOS (NLOS) makes it possible to fairly compare analytical results with Monte-Carlo simulations for SM, which strongly depend on the rank of the channel. Numerical results reveal the following insights: (i)  In interference-limited scenarios, $\SINR$ coverage has a non-monotonic trend with BS density. The optimum BS density for $\SINR$ coverage decreases with increasing degree of multiuser transmission. (ii) Although $\SINR$ coverage decreases with MU-MIMO, the median and peak per user rate increases due to increasing number of time slots available per user. However, the cell edge rates suffer with round robin scheduling, which motivates that the scheduler must explicitly safeguard the rates of edge users to use MU-MIMO. 

\subsubsection{Comparison of MIMO Techniques Considering Coverage, Rate and Power Consumption Tradeoffs}
With perfect channel state information at the transmitter and neglecting channel acquisition and computational complexity overheads, MU-MIMO usually provides higher per user throughput compared to SM and SU-BF in mmWave networks for a fixed density of BSs and fixed number of antennas per BS/UE. Further note that enabling MU-MIMO requires only single RF chain at UEs, whereas enabling SM requires some baseband combining at UEs with multiple RF chains. We provide a stochastic ordering argument which highlights that $\SNR$ coverage normalized by the antenna gains is better for MU-MIMO than SM asymptotically with the number of antennas at the BSs and users. SM can outperform MU-MIMO in scenarios when SM can support more streams than the number of users that can be served with MU-MIMO. This boils down to having sufficiently low user density coupled with sufficiently large number of RF chains at UEs/BSs and multipath in the channel.
Instead of fixing the density of BSs if power consumption per unit area is fixed, a denser SU-BF network outperforms MU-MIMO and SM in terms of per user cell edge rates. However, the sum rate with MU-MIMO is still usually better than SU-BF and SM. The above results on sum or per user rates neglect the possibly increased overheads with MU-MIMO due to channel acquisition or computational complexity. Incorporating such factors, our results can be re-interpreted so as to quantify the minimum allowable efficiency for MU-MIMO to provide higher data rates than SM or SU-BF. The definition of minimum allowable efficiency is formally given in Section~\ref{sec:analysis}.
\subsection{Organization and Notation}
Section~\ref{sec:model} sets up the system model. The analytical model for coverage and rate in MU-MIMO mmWave networks is developed in Section~\ref{sec:MUMIMO}. Heuristic comparison of coverage and rate with SM is discussed in Section~\ref{sec:SM}. Section~\ref{sec:results} and \ref{sec:conclusion} discusses the numerical results and conclusions.\footnote{Variables in italics are scalar random variables. Small and capital bold letters indicate vectors and matrices, respectively. An exception are random spatial locations in $\mathbb{R}^2$, which are italicized small letters $\x,y,\us,v$ or $w$. The complex conjugate transpose and pseudo inverse of ${\bf A}$ is ${\bf A}^*$ and ${\bf A^{\dagger}}$, respectively.  Convergence in distribution is denoted by $\stackrel{d}{\to}$.}.  %1 page
% !TeX spellcheck = <none>
\section{System Model}
\label{sec:model}
Consider a downlink mmWave cellular network operating at carrier frequency $\mathrm{f}_c$ and with bandwidth $\res$. It is assumed that BSs and UEs are distributed in $\mathbb{R}^2$ as independent and homogeneous Poisson point processes (PPPs) $\PPP$ and $\PPPu$, with intensities $\dnsty$ and $\userdnsty$, respectively\cite{andganbac11}. Each BS and user is assumed to employ a uniform linear array (ULA) of size $\NBS$ and $\NMS$, respectively. Full buffer traffic is assumed in this work. %This assumption may be relaxed in future work by considering a load aware model as done in  \cite{DhiGanJ2013}. 

\subsection{Propagation Model}
Path loss from BS at $\x\in\PPP$ to a user at $\us\in\PPPu$ is given in dB by 
\begin{equation}
L(\x,\us) = \beta + 10\alpha \log_{10}(||\x-\us||) + S_{\x,\us},
\end{equation}
where $\beta= 20\log_{10}\left(\frac{4\pi}{\lambda_c}\right)$ is the reference distance path loss at 1 meter, $\lambda_c$ is the wavelength in meters, $\alpha$ is the path loss exponent, $S_{\x,\us} \sim \mathit{N}(0,\xi^2)$ denotes Gaussian distribution with zero mean and standard deviation $\xi$. Note that $\alpha$ and $\xi$ are different for LOS and NLOS links. A subscript `$\mathrm{L}$' and `$\mathrm{N}$' to  $\alpha$ and $\xi$ denote the respective parameters for LOS and NLOS links, respectively. A probabilistic blockage model proposed and validated in \cite{SinJSAC14,Kul14} is used in this work. According to this model, the probability that a link of length $||x-u||$ is LOS is $\plos$ if $||x-u||\leq \dmax$, for some value of $\dmax$. All links longer than $\dmax$ are NLOS. 

MmWave channels are expected to be sparse with very few angles of arrival (AOAs) and departure (AODs) capturing most of the energy \cite{Akd14,Sam16,SamRap14,Sun14}. In this work, we assume a narrowband geometric channel model\cite{Aya14,Alk14}, where the channel matrix between BS at $\x$ and user $\us$ is given by 
\begin{equation}
\label{eq:channel1}
{\bf H}_{\x,\us} = \sqrt{\frac{\NBS\NMS}{L(\x,\us)\eta_{\x,\us}}}\sum_{i=1}^{\eta_{\x,\us}}\gamma_{i,\x,\us} \aMS(\AOA{i}{\x}{\us}) \mathrm{\bf a}^{*}_\mathrm{BS}(\AOD{i}{\x}{\us}).
\end{equation}
Here, $\eta_{\x,\us}$ is the number of paths between BS at $\x$ and user at $\us$, 
$\gamma_{i,\x,\us}$ is the small scale fading on $i^{\text{th}}$ path (assumed to be complex normal with zero mean and unit variance for both LOS and NLOS to enhance analytical tractability), $\AOD{i}{\x}{\us}$ is the virtual AOD and $\AOA{i}{\x}{\us}$ is the virtual AOA for the $i^{\text{th}}$ path. The number of paths $\eta_{\x,\us}$ equals $\eta_\mathrm{L}$ or $\eta_\mathrm{N}$ depending on whether the link is LOS or NLOS, respectively\footnote{$\eta_\mathrm{L}>1$ indicates more than 1 LOS {\em like} paths. In this work, we either have LOS or NLOS multipaths. A more general channel model would incorporate scenarios with 1 or more LOS like paths along with NLOS paths. However, an optimal power allocation would nearly allocate all power to LOS-like paths, thus, justifying our model. For simplicity, it is assumed that each scatterer gives rise to a single dominant path \cite{Say02,Alk14,Aya12}. Extension to a clustered model\cite{Aya14,Say02} is desirable in future.}.
%depends on whether the link is LOS or NLOS. 
%\begin{equation*}
%\eta_{\x,\us} = \begin{cases}
%\eta_\mathrm{L} & \text{if link is LOS}\\
%\eta_\mathrm{N} & \text{otherwise.}\\
%\end{cases}
%\end{equation*}
It is expected that $\eta_\mathrm{N}>\eta_\mathrm{L}$\cite{Sun14,Sam16,SamRap14}. The virtual AOA or AOD are related to the corresponding physical angles as 
%\begin{equation*}
%\label{eq:virtualphymap}
$\theta = 2\pi \text{d} \sin(\varphi)/\lambda_c$,
%\end{equation*}
where d is the inter-antenna spacing (chosen to be $\lambda_c/2$), $\varphi$ is the physical angle and $\theta$ is the virtual angle. The array response vectors for ULAs, $\aBS$ and $\aMS$, are of the form
%\begin{equation*}
%\label{eq:channel}
$\mathrm{\bf a}(\theta) =  [1\; e^{-j \theta} \ldots e^{-j (\text{N}-1)\theta}]^{*}/\sqrt{\mathrm{N}}$,
%\end{equation*}
where $\mathrm{N}\in\{\NBS,\NMS\}$. We assume that for every BS-UE link, scatterers in environment are uniformly distributed in $[0,2\pi]$ and thus, the physical angles are also uniformly distributed in $[0,2\pi]$. We call this the ``physical model'', which will be the basis of our  simulation results, whereas for tractable analysis we leverage the virtual channel approximation\cite{Say02} in Section~\ref{sec:analysis}.
%From \cite{Alk14}, we expect such an approximation to be tight in {\em large dimension regime}, that is when the number of antennas at the BSs and UEs is large.

\subsection{Fully Connected Hybrid Beamforming Architecture}
A fully-connected two layer hybrid beamforming architecture with $\mathrm{N}^{\mathrm{BS}}_\mathrm{RF}$ and $\mathrm{N}^{\mathrm{UE}}_\mathrm{RF}$ RF chains at the BS and UE, respectively, is shown in Figure~\ref{fig:hybridarchitecture}. A BS at $\x$ sends a total of $\mathrm{N}^\mathrm{BS}_\mathrm{s}$ streams of data, which may include data sent to multiple users in the network. The transmit signals first go through a $\mathrm{N}^\mathrm{BS}_\mathrm{RF} \times \mathrm{N}^\mathrm{BS}_\mathrm{s}$ baseband precoder matrix $\bbprecoder = [\bbcolumn{1}{\x}\ldots \bbcolumn{\mathrm{N}^\mathrm{BS}_\mathrm{s}}{\x}]$ followed by a $\mathrm{N}_\mathrm{BS} \times \mathrm{N}^\mathrm{BS}_\mathrm{RF}$ RF precoder $\rfprecoder = [\rfcolumn{1}{\x}\ldots \rfcolumn{\mathrm{N}^{\mathrm{BS}}_\mathrm{RF}}{\x}]$. Note that the RF precoder is generally implemented using phase-shifters \cite{Aya14,Alk14}, although there have been attempts trying to explore alternative methods\cite{rial15}. 
Let us denote the RF combiner at user $u$ by $\urfprecoder$ and the baseband combiner by $\ubbprecoder = [\ubbcolumn{1}{\us},\ldots,\ubbcolumn{\mathrm{N}^{\mathrm{UE}}_\mathrm{s}}{\us}]$. Note that SM, MU-MIMO and SU-BF can all be implemented with this generic architecture. The problem of jointly optimizing over $\rfprecoder$, $\bbprecoder$, $\urfprecoder$ and $\ubbprecoder$ to maximize sum rate or per user rate for SM and MU-MIMO is still an open problem\cite{Alk14,Aya14}. In the following sections, we  use recently proposed near optimal algorithms for designing of precoders and combiners in \cite{Aya14} and \cite{Alk14} to employ SM and MU-MIMO, respectively as baseline for our simulations and analysis. 
\ifSC
\begin{figure}
 	\centering
 		\includegraphics[width=\columnwidth]{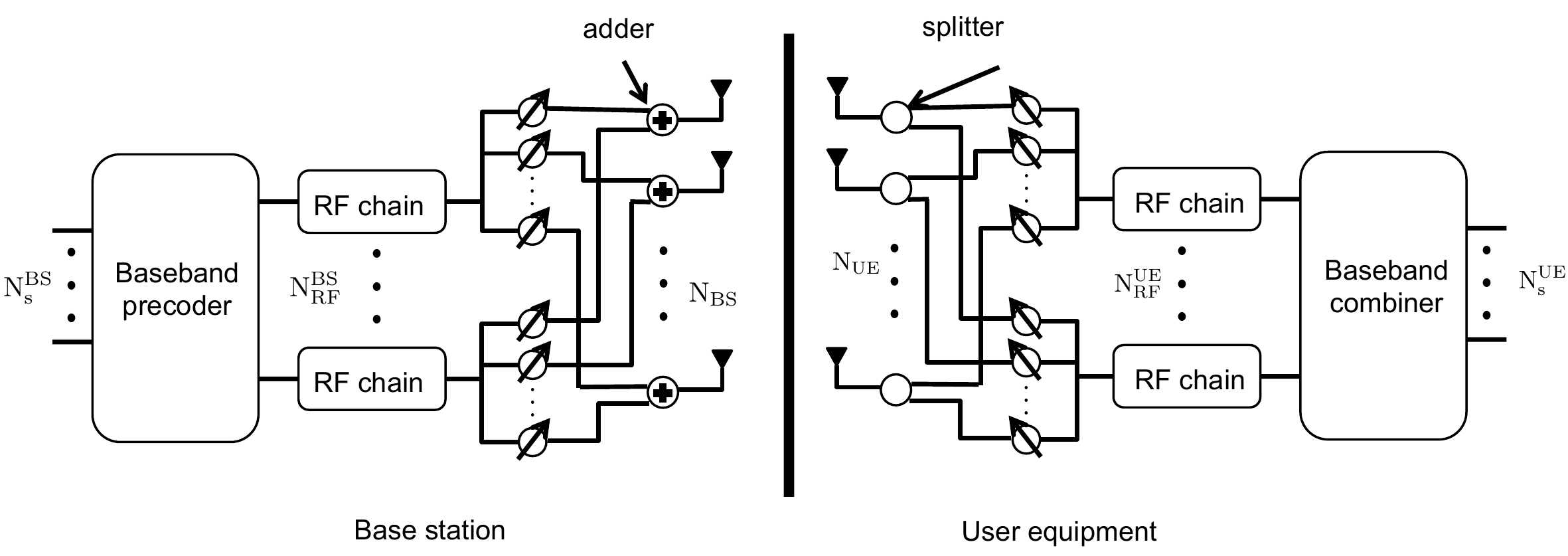}
 		\caption{Fully-connected hybrid architecture at the BSs and UEs.}
 	\label{fig:hybridarchitecture}
 \end{figure} 
\else
\begin{figure*}
 	\centering
 		\includegraphics[width=1.9\columnwidth]{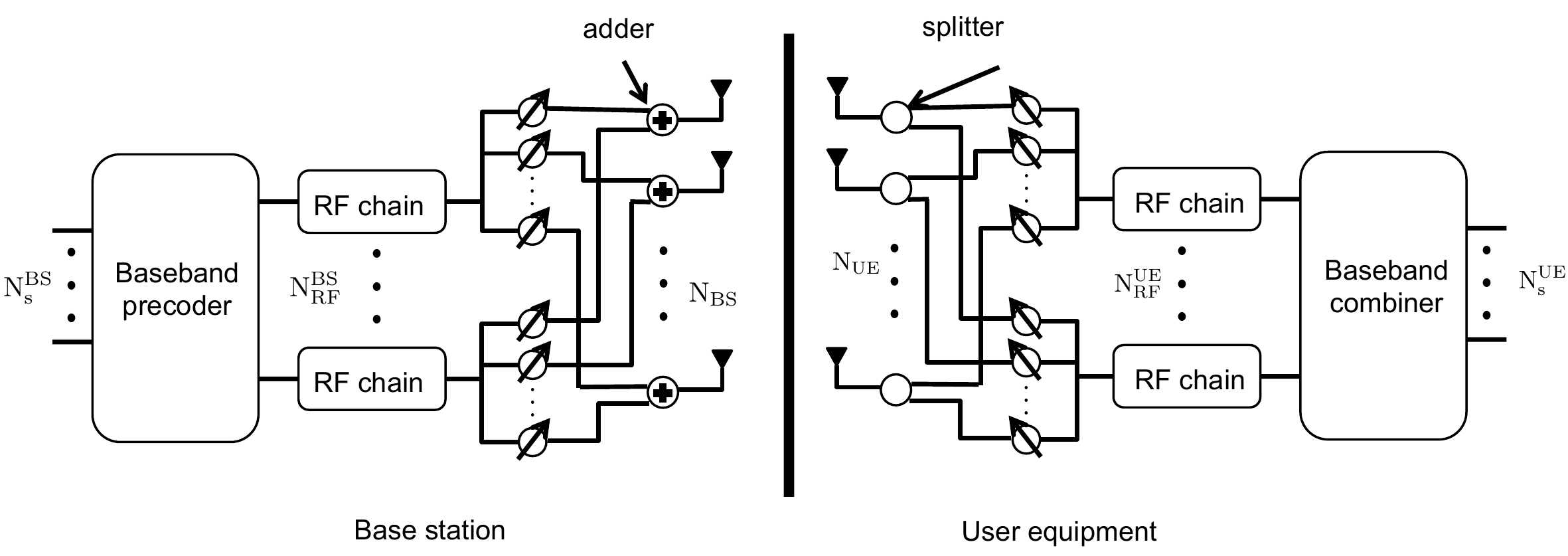}
 		\caption{Fully-connected hybrid architecture at the BSs and UEs.}
 	\label{fig:hybridarchitecture}
 \end{figure*} 
\fi
%Such an architecture was proposed in \cite{Vook14, Aya14} and later used for analysis of limited feedback multiuser mmWave cellular networks in \cite{Alk14}. 

  %1.5 page
\section{Multiuser MIMO in mmWave Cellular Networks}
\label{sec:MUMIMO}
For MU-MIMO, it is assumed that each BS serves multiple users with a single stream per user. Thus, analog beamforming with a single RF chain suffices at each UE. 
%The block diagram in Fig.~\ref{fig:hybridarchitecture} simplifies to Fig.~\ref{fig:mumimo}. 
Let $\mathcal{U}_\x$ be the set of all users in $\PPPu$ which are scheduled by the BS at $\x$ in the same time slot, and the cardinality of $\mathcal{U}_\x$ be $U_\x$. We assume $U_\x = \min(\NRF,N_\x)$, where $N_\x$ is the total number of users connected to the BS and $\NRF$ is the maximum number of users that can be scheduled in a time slot. A more sophisticated algorithm for deciding how many and which users to schedule in a resource block may be implemented as in \cite{Adh14,Voo14} but we neglect this aspect here for tractability. We further assume $\NRF = \mathrm{N}_\mathrm{RF}$, and that unless the load on the BS is less than the number of RF chains, $\NRF$ users are served in a time slot. Also, when $\NRF>U_\x$ only $U_\x$ RF chains are used for processing, which means that $\bbprecoder= [\bbcolumn{1}{\x}\ldots \bbcolumn{U_\x}{\x}]$ is of dimension $U_\x \times U_\x$ and $\rfprecoder=[\rfcolumn{1}{\x}\ldots \rfcolumn{U_\x}{\x}]$ is of dimension $\NBS \times U_\x$.
%\begin{figure}
% 	\centering
% 		\includegraphics[width=\columnwidth]{MUMIMO}
% 		\caption{MU-MIMO system under consideration.}
% 	\label{fig:mumimo}
 %\end{figure}

Under the narrowband assumption, the received signal at user $\us$ from BS at $\x$ after passing through ${\bf w}_\us$, the RF combiner at the user, is given by
\begin{equation*}
y_\us = \heffective \bbcolumn{\us}{\x} s_\us + \sum_{{v}\in\mathcal{U}_\x, {v}\neq u} \heffective \bbcolumn{{ v}}{\x} s_{{v}} + \mathrm{OCI} + \text{noise}, 
\end{equation*}
where $\heffective = {\bf w}^{*}_\us {\bf H}_{\x,\us} \rfprecoder$ and OCI is the out-of-cell interference. Here, $s_{(.)}$ are the transmit symbols with energy $\power{}/U_\x$. Thus, the total transmit power of the BS is  $\power{}$. In this work, we will follow the precoding/combining algorithm in \cite{Alk14} and also assume an infinite resolution codebook at BSs and UEs for tractability. The first step is to design the RF precoders and combiners that maximize the received signal power on each of the BS-UE links. Thus, ${\bf w}_\us$ and $\rfcolumn{\us}{\x}$ are designed such that $
({\bf w}_\us, \rfcolumn{\us}{\x}) = \underset{{\bf w,\;f}}{\arg \max}|{\bf w}^{*}{\bf H}_{\x,\us}{\bf f}|.
$

\begin{lem}[from \cite{Aya12}]
\label{lem:precodcombin}
The left and right singular vectors corresponding to non-zero eigenvalues  of ${\bf H}_{\x,\us}$ with $\eta_{\x,\us}\ll\min(\NBS,\NMS)$ converge in chordal distance to $\mathrm{\bf a}_\mathrm{UE}(\AOA{i}{\x}{\us})$ and $\mathrm{\bf a}_\mathrm{BS}(\AOD{i}{\x}{\us})$, for $1\leq i\leq \eta_{\x,\us}$. The corresponding singular values converge to $\frac{\NBS\NMS}{L(\x,\us)\eta_{\x,\us}}|\gamma_{i,\x,\us}|^2$.
\end{lem}
This lemma indicates that for large number of antennas 
${\bf w}_\us = \mathrm{\bf a}_\mathrm{UE}(\AOA{i_m}{\x}{\us})$ and $\rfcolumn{\us}{\x} = \mathrm{\bf a}_\mathrm{BS}(\AOD{i_m}{\x}{\us})$, where $i_m = \underset{i}{\arg \max} |\gamma_{i,\x,\us}|$. This observation will be crucial in developing a tractable model for coverage and rate. Next, the baseband precoder is designed such that the multiuser interference is cancelled.
Using a zero forcing (ZF) baseband precoder, 
$\bbprecoder = \overline{\bf H}^{\dagger}_{\x} {\bf \Lambda}
$,
where $\bf \Lambda$ is a diagonal matrix whose entries are chosen such that $||\rfprecoder \bbcolumn{\us}{\x}|| = 1$. Here,  $\overline{\bf H}_{\x} = [\overline{\bf h}_{\x,\us_1}\ldots \overline{\bf h}_{\x,\us_{U_\x}}]^{*}$ with $\mathcal{U}_\x = \{\us_1,\ldots,\us_{U_\x}\}$. Note that 
 $\overline{\bf H}^{\dagger} = \overline{\bf H}^{*} \left(\overline{\bf H} \overline{\bf H}^{*}\right)^{-1}$, if $\overline{\bf H}$ is full rank.

\subsection{$\SINR$ and Rate Model}
The $\SINR$ of the user at $\us\in\PPPu$ served by a BS at $\x\in\PPP$ connected to $U_\x$ total users is given by 
\begin{equation}
\label{eq:sinr1}
\SINR_{\x,\us} = \frac{\frac{ ||\heffective \bbcolumn{\us}{\x}||^2}{U_\x}}{\frac{\sigma^2_n}{\power{}}+  \sum\limits_{\substack{v\in\mathcal{U}_x
\\{v}\neq \us}}\frac{ ||\heffective \bbcolumn{v}{\x}||^2 }{U_\x}+ \sum\limits_{\substack{y\in\PPP\\y\neq \x}} \sum\limits_{w\in\mathcal{U}_y} \frac{||\overline{\bf h}^*_{y,\us} \bbcolumn{w}{y}||^2}{U_y}}.
\end{equation}
The second term in the denominator is zero, owing to the ZF precoder and the fact that $\overline{\bf H}_\x$ is almost surely full rank for independently distributed channel gains from BS at $\x$ to users in $\mathcal{U}_\x$. %Note that in reality it may not possible to obtain narrowband channel even with schemes like SC-FDM (single-carrier frequency division multiplexing) or OFDM (orthogonal frequency division multiplexing) due to the very large bandwidth of mmWave networks.  
The per user rate (in bits per second or bps) of user $u$ served by BS at $\x$ is defined as
\begin{equation}
\label{eq:peruserrate}
R_{\x,\us} = \omega_\x \frac{\res\, U_\x}{N_\x} \log_2(1+\SINR_{\x,\us}),
\end{equation}
where $\omega_\x < 1$ models the efficiency in implementing MU-MIMO in terms channel acquisition or computational complexity or cyclic prefix while implementing OFDM \cite{Bai2015,Marzetta10}. The above model implies that each user gets $U_\x/N_\x$ fraction of resources, which can be achieved using round robin scheduling. The sum rate is defined as
\begin{equation}
R_{\x} = \omega_\x \res \sum_{u\in\mathcal{U}_x} \log_2(1+\SINR_{\x,\us}),
\end{equation}
which is basically the total number of bits per second (bps) transmitted by the BS, whereas the per user rate is the rate achieved by a typical user in a scheduling cycle. 

In general, the efficiency factors vary for different BSs and are dependent on $\NRF,\NBS,\NMS$,$\eta_\mathrm{N}$, $\eta_\mathrm{L}$ and OFDM cyclic prefix penalty. For simplicity we assume $\omega_\x = \omega_\mathrm{MU}, \forall \x\in\PPP $. One can interpret $\omega_\mathrm{MU} = \min_\x \omega_\x$ to get a lower bound on the rate. Since we expect the overhead to increase with $U_\x$, $\omega_\mathrm{MU}$ corresponds to the efficiency of BSs serving $\NRF$ users. Note that $\omega_\mathrm{MU}$  for $\NRF=1$ is the overhead for SU-BF.  
%Both of these metrics are essential to study the benefits of MU-MIMO. Although the sum rate helps to capture the fact that the BSs are in fact serving multiple users, it tends to hide the information about rates achieved by the cell edge users. This can be better studied using the per user rate as the metric. In the following section, we will discuss the analytical model for $\SINR$ and per user rate distribution based on the virtual channel model approximation. Analysing the distribution of sum rate requires to have a joint distribution of $\SINR$ of all users scheduled in a time slot, which is not tractable.

\subsection{Coverage and Rate Analysis}
\label{sec:analysis}
Consider a {\em typical} UE at origin, wherein the notion of typicality for stationary point process is defined through Palm probability \cite{BacBook09}, and it associates with the BS at $\x$ offering minimum path loss $L(\x,0)$. We call this the {\em tagged} BS. We evaluate the $\SINR$ coverage defined as $\mathbb{P}\left(\SINR_{x,0}>\tau\right)$, which is the $\SINR$ distribution of the {\em typical} user at origin. Rate coverage is similarly defined.  The $\SINR$ expression in (\ref{eq:sinr1}), although exact, is not tractable in terms of finding its distribution. We, thus, provide an accurate yet tractable approximation that captures the dependency of the several parameters in the following analysis. %The notion of stochastic ordering will be useful, and thus we provide its definition here. 
\begin{definition}
A random variable $\mathrm{Z}_1$ stochastically dominates another random variable $\mathrm{Z}_2$, if $\pr(Z_1>z)\geq \pr(Z_2>z)$ for all $z\in\mathbb{R}$. We denote this as $Z_1\stackrel{\text{st}}{\geq}Z_2$.
\end{definition}
\subsubsection{Rate Distribution in a Noise-limited Network}
We first focus on finding the rate distribution in a network with negligible interference effects. Throughout the discussion, the virtual angles of departure/arrival are quantized to take values in $\{\theta: \theta = -\pi+\frac{2\pi i}{\mathrm{N}_a},  1\leq i \leq \mathrm{N}_a\}$.
\begin{lem}
\label{lem:angle}
If antenna spacing is half wavelength and the physical AOAs/AODs are uniformly distributed in 0 to $2\pi$, the distribution of the quantized virtual angles is given by
\begin{equation*}
\ifSC
q_{a,i} = \left(\arcsin\left(-1+\frac{2i+1}{\mathrm{N}_a}\right)-\arcsin\left(-1+\frac{2i-1}{\mathrm{N}_a}\right)\right)/\pi,
\else
q_{a,i} = \frac{\left(\sin^{-1}\left(-1+\frac{2i+1}{\mathrm{N}_a}\right)-\sin^{-1}\left(-1+\frac{2i-1}{\mathrm{N}_a}\right)\right)}{\pi},
\fi
\end{equation*}
for $a \in \{\mathrm{UE}, \mathrm{BS}\}$ and $i\in\{1,\ldots,\mathrm{N}_a-1\}$.  Further, $q_{a,\mathrm{N}_a} = 1-\sum_{j=1}^{\mathrm{N}_a-1}q_{a,j}$.
\end{lem}
\begin{proof}
Note that $\theta = \pi \sin(\varphi)$ for half wavelength antenna spacing. 
%Thus, 
%\begin{equation*}
%q_{a,i} = 
%\pr\left(\pi \sin(\varphi)\in\left[-\pi + \frac{\pi (2i-1) }{\mathrm{N}_a}, -\pi + \frac{\pi (2i+1) }{\mathrm{N}_a}\right)\right).
%\end{equation*}
Thus, the required probability can be computed by using that $\varphi$ is uniformly distributed in 0 to $2\pi$. 
\end{proof}
%\begin{lem}
%\label{lem:virtual}
%The probability that all paths between a BS-user link are {\em distinct} goes to 1 as number of antennas goes to infinity.
%\end{lem}
%\begin{proof}
%If $\eta$ is the number of paths between a BS and a user. The required probability is equal to $\prod_{i=1}^{\eta}(1-q_{\mathrm{UE},i})\prod_{j=1}^{\eta}(1-j/\NBS)$, which approaches 1 as $\NBS$ and $\NMS$ go to infinity for fixed $\eta$.
%\end{proof}
%We will need this Lemma for justifying Proposition~\ref{prop:1} (mentioned next).
\begin{prop}
\label{prop:1}
Let $\mathcal{U}_\x = \{\us_1,\ldots,\us_{\mathrm{U_\x}}\}$ be the users served by the BS at $\x$. Assuming $\eta_N \ll \min(\NBS,\NMS)$, $\NRF \ll \min(\NBS,\NMS)$ and a dense network deployment, $\SNR$ at user $\us_1$ can be modelled as 
\begin{equation}
\label{eq:sinr2}
\SNR_{\x,\us_1} \approx \frac{\mathrm{G}}{\eta_{\x,\us_1} U_\x \sigma^2_n} |\gamma_{i_m,\x,\us_1}|^2  L(\x,\us_1)^{-1} p_{\mathrm{ZF}},
\end{equation}
where $\mathrm{G} = \power{} \NBS \NMS$, $i_m$ is the index corresponding to $\underset{i}{\arg\max}|\gamma_{i,\x,\us_1}|$, $p_{\mathrm{ZF}}$ is a random variable that captures reduction in signal power due to the ZF penalty and has distribution that stochastically dominates $p_{\mathrm{MU}}$, which is a Bernoulli random variable with success probability $\zeta(\eta_{\x,\us_1},U_\x)$,
%\begin{equation*}
%\label{eq:ZFpenalty}
%\ifSC
%p_{\mathrm{MU}} = \begin{cases}
%1 & \text{ w. p. } \zeta(\eta_{\x,\us_1},U_\x) = \sum_{j=1}^{\NBS} q_{\BS,j}\mathrm{B}_j(\eta_{\x,\us_1},U_\x) \left(\mathrm{p}_{\LOS} \mathrm{A}_j(\eta_\mathrm{L}) + (1-\mathrm{p}_\LOS)\mathrm{A}_j(\eta_\mathrm{N}) \right)^{U_\x-1}\\
%0 & \text{ otherwise,}
%\end{cases}
%\else
%p_{\mathrm{MU}} = \begin{cases}
%1 & \text{ w. p. } \zeta(\eta_{\x,\us_1},U_\x)\\
%0 & \text{ otherwise,}
%\end{cases}
%\fi
%\end{equation*}
%\begin{equation*}
%,
%\end{equation*}
where 
\ifSC
$
\zeta(\eta_{\x,\us_1},U_\x) = \sum_{j=1}^{\NBS} q_{\BS,j}\mathrm{B}_j(\eta_{\x,\us_1},U_\x) \left(\mathrm{p}_{\LOS} \mathrm{A}_j(\eta_\mathrm{L})\right. \left.+ (1-\mathrm{p}_\LOS)\mathrm{A}_j(\eta_\mathrm{N}) \right)^{U_\x-1},
$
\else
\begin{multline*}
\zeta(\eta_{\x,\us_1},U_\x) = \sum_{j=1}^{\NBS} q_{\BS,j}\mathrm{B}_j(\eta_{\x,\us_1},U_\x) \left(\mathrm{p}_{\LOS} \mathrm{A}_j(\eta_\mathrm{L})\right.\\ \left.+ (1-\mathrm{p}_\LOS)\mathrm{A}_j(\eta_\mathrm{N}) \right)^{U_\x-1},
\end{multline*}
\fi
$\mathrm{A}_j(\eta) = \mathrm{C}(\eta) + (1-q_{\BS,j})^{\eta-1} -(1-q_{\BS,j})^{\eta-1} \mathrm{C}(\eta) $, $\mathrm{B}_j(\eta, ,U_\x) = \mathrm{C}(\eta)(1-q_{\BS,j})^{U_\x-1} +  \mathrm{D}_j(\eta,U_\x) - \mathrm{D}_j(\eta,U_\x)\mathrm{C}(\eta)$, $\mathrm{C}(\eta) = \sum_{i=1}^{\NMS} q_{\UE,i} (1-q_{\UE,i})^{\eta-1}$, and 
\begin{align*}
\iftoggle{SC}{}{\hspace{-0.5cm}}\mathrm{D}_j(\eta,U_\x)= \sum\limits_{i_1,\ldots,i_{\eta-1}=1}^{\NBS-1} \prod\limits_{n=1}^{\eta-1} l_{i_n,j} \left(1-\sum\limits_{\text{unique}(i_{(.)})} l_{i_n,j}\right)^{U_\x-1},
\end{align*}
where $l_{n,j} =\frac{q_{\BS,n}}{1-q_{\BS,j}}$ if $n<j$ and $l_{n,j} = \frac{q_{\BS,n+1}}{1-q_{\BS,j}}$ if $n\geq j$ and unique$(i_{(.)})$ represents the unique values in the set $\{i_1,\ldots,i_{\eta-1}\}$. 
\end{prop}
\begin{proof}
Without loss of generality, $i_m=1$. From Lemma~\ref{lem:precodcombin}, ${\bf w}_\us = \aMS(\AOA{1}{\x}{\us})$ and $\rfcolumn{\us}{\x} = \aBS(\AOD{1}{\x}{\us})$. Using the orthogonality of the array response vectors with quantized virtual angles, it can be shown that $\overline{\bf H}_\x$ takes the form
 \begin{equation}
 \label{eq:equivalentchannel}
 \overline{\bf H}_\x= \left[ \begin{array}{cc}
 \sqrt{\frac{\NBS\NMS}{L(\x,\us_1)\eta_{\x,\us_1}}}\gamma_{1,\x,\us_1}\textbf{}&{\bf 0}\\
 {\bf 0}&\tilde{\bf P}_\x
  \end{array}\right],
 \end{equation}
with probability at least $\zeta(\eta_{x,\us_1},U_\x)$. See Appendix~\ref{app:prop1} for details. Note that here $\tilde{\bf P}_\x$ is a submatrix of $\overline{\bf H}_\x$ of dimension $U_\x-1 \times U_\x-1$. In this case,
\begin{equation*}
\overline{{\bf H}}_\x^\dagger =\left[ \begin{array}{cc}
\sqrt{\frac{L(\x,\us_1)\eta_{\x,\us_1}}{\NBS\NMS}}\gamma^{-1}_{1,\x,\us_1}&{\bf 0}\\
 {\bf 0}&\tilde{\bf P}^{\dagger}_\x
  \end{array}\right].
\end{equation*}
We know that $\bbprecoder = \overline{\bf H}_\x^\dagger {\bf \Lambda}$, for diagonal matrix ${\bf \Lambda}$ that helps satisfy the power constraints. Thus, 
the first column of the baseband precoder is of the form $\bbcolumn{\us_1}{\x} = [c\; 0\ldots0]$, for some constant $c$ such that $||\rfprecoder \bbcolumn{\us_1}{\x}||=1$. Thus, 
$\bbcolumn{\us_1}{\x} = [1\; 0\ldots0]$ since each term in $\rfprecoder$ is unit norm. In this case, the received signal power of $\us_1$ is equal to 
$\frac{G}{\eta_{\x,0}U_\x}|\gamma_{1,\x,\us_1}|^2 L(\x,\us_1)^{-1}$, which corresponds to the case when $p_{\mathrm{MU}}=1$ in (\ref{eq:sinr2}). Since the event that ${\bf P}_\x$ is not of this form is of low probability and results in intractable expressions the signal power is lower bounded by 0 in this case. Under virtual channel approximation, (\ref{eq:sinr2}) is a lower bound on $\SNR$. 
% Since this event occurs with very low probability for large $\NBS$, we lower bound the received signal power to 0 in this case by choosing $p_{\mathrm{MU}}=0$. % From this discussion it is clear that by assuming the virtual channel approximation, the $\SNR$ charecterization in (\ref{eq:sinr2}) is a lower bound on the actual $\SNR$.
\end{proof}
\begin{rem}
\label{rem:simple}
If the quantized virtual angles are distributed uniformly in their range, instead of the distribution in Lemma~\ref{lem:angle},  %$\zeta(\eta_{x,\us_1},U_\x)$ 
$\mathrm{D}_j(\eta, U_\x)$ takes a much simplified form given by 
%\ifSC
%\begin{equation*}
%\zeta(\eta_{x,\us_1},U_\x) = \mathrm{B}(\eta_{x,\us_1},U_\x) \left(\mathrm{p}_{\LOS} \mathrm{A}(\eta_\mathrm{L}) + (1-\mathrm{p}_\LOS)\mathrm{A}(\eta_\mathrm{N}) \right)^{U_\x-1},
%\end{equation*} 
%where $\mathrm{B}(\eta,U_\x) = (1-1/\NMS)^{U_x-1} (1-1/\NBS)^{U_x-1} + \mathrm{D}(\eta,U_\x)- \mathrm{D}(\eta,U_\x)(1-1/\NMS)^{U_x-1}$,
%$\mathrm{D}(\eta,U_\x) = \sum_{d=1}^{\eta-1}{\NBS-1 \choose d} (\NBS-1-d)^{U_\x-1}\sum_{i=0}^{d} (-1)^{i} (d-i)^{\eta-1}{d \choose i}$ and $\mathrm{A}(\eta)=(1-1/\NMS)^{\eta-1} + (1-1/\NBS)^{\eta-1} -(1-1/\NBS)^{\eta-1} (1-1/\NMS)^{\eta-1}$. 
%\else
%\begin{multline*}
%\zeta(\eta_{x,\us_1},U_\x) = \mathrm{B}(\eta_{x,\us_1},U_\x) \left(\mathrm{p}_{\LOS} \mathrm{A}(\eta_\mathrm{L}) \right. \\
%\left.+ (1-\mathrm{p}_\LOS)\mathrm{A}(\eta_\mathrm{N}) \right)^{U_\x-1},
%\end{multline*}
%\begin{multline*}
%\mathrm{B}(\eta,U_\x) = (1-1/\NMS)^{U_x-1} (1-1/\NBS)^{U_x-1} + \\
%\mathrm{D}(\eta,U_\x)- \mathrm{D}(\eta,U_\x)(1-1/\NMS)^{U_x-1}
%\end{multline*}
 %
 \ifSC
$\mathrm{D}(\eta,U_\x) = \sum_{d=1}^{\eta-1}{\NBS-1 \choose d} (\NBS-1-d)^{U_\x-1}\sum_{i=0}^{d} (-1)^{i} (d-i)^{\eta-1}{d \choose i} .$ See Appendix~\ref{app:prop1} for details. 
\else
$
\mathrm{D}(\eta,U_\x) = \sum_{d=1}^{\eta-1}{\NBS-1 \choose d} (\NBS-1-d)^{U_\x-1}\\\times\sum_{i=0}^{d} (-1)^{i} (d-i)^{\eta-1}{d \choose i} .%, \text{ and }
$
\fi
%\begin{multline*}
%\mathrm{A}(\eta)=(1-1/\NMS)^{\eta-1} + (1-1/\NBS)^{\eta-1} \\ - (1-1/\NBS)^{\eta-1} (1-1/\NMS)^{\eta-1}.
%\end{multline*} 
%\fi
\end{rem}
\begin{rem}
To simplify evaluation of Proposition~\ref{prop:1}, the following can be used  

$(1-q_{\mathrm{BS},j})^{U_\x-1} \sum_{i=1}^{\NMS}q_{\mathrm{UE},i}(1-q_{\mathrm{UE},i})^{\eta-1} \leq \mathrm{B}_{j}(\eta,U_\x)\leq (1-q_{\mathrm{BS},j})^{U_\x-1}$.
\end{rem}

\begin{rem}
\label{rem:convergence}
It can be shown that   $\sum_{i=1}^{\NBS}q_{\mathrm{BS},i}(1-q_{\mathrm{BS},i})^{r} \to 1$
as $\NBS\to \infty$ for any $r\geq 0$, which is true since $\max_i q_{\mathrm{BS},i}\to 0$ as $\NBS\to \infty$. Similar result holds for $q_{\mathrm{UE},j}$ with $\NMS\to \infty$. All these imply that $\zeta\to 1$ with $\NBS\to \infty$ and $\NMS\to \infty$. 
\end{rem}

To find the $\SNR$ coverage, we need to find the distribution of several random variables in Proposition~\ref{prop:1}. First we focus on the finding the probability mass function (PMF) of the number of multiuser streams of BS at $y\in\PPP$ given by $U_y = \min\{\NRF,N_y\}$. We use an approximation proposed in \cite{SinDhiAnd13} to model the distribution of $N_y$, which are actually correlated random variables for $y\in\PPP$ and particularly known to be intractable since finding the volume of Voronoi association cells is itself an unsolved problem\cite{Ferenc2007}. With notably different propagation channels for LOS and NLOS links, the cell association regions in mmWave networks are not even Voronoi and more irregular \cite{SinJSAC14}. The PMF of $N_y$ is denoted by $\kappa(n)$ is modelled as follows\cite{SinJSAC14}. Let $\rho = \userdnsty/\dnsty$, then if $y=\x$, that is the BS is serving the typical user, 
\ifSC
\begin{equation}
\label{eq:kappa}
\kappa(n)=\frac{3.5^{3.5}}{(n-1)!}\frac{\Gamma(n+3.5)}{\Gamma(3.5)}\rho^{n-1} \left(3.5+\rho\right)^{-n-3.5}, \text{  for $n\geq 1$ and $\kappa(0)=0$.}
\end{equation}
For interfering BSs, 
\begin{equation}
\label{eq:kappa2} \kappa(n)=\frac{3.5^{3.5}}{n!}\frac{\Gamma(n+3.5)}{\Gamma(3.5)}\rho^{n} \left(3.5+\rho \right)^{-n-3.5},  \text{   for $n\geq 0$.}
\end{equation}
\else
$\kappa(n) $ is given by 
\begin{equation}
\label{eq:kappa}
\frac{3.5^{3.5}}{(n-1)!}\frac{\Gamma(n+3.5)}{\Gamma(3.5)}\rho^{n-1} \left(3.5+\rho\right)^{-n-3.5},
\end{equation}
for $n\geq 1$ and $\kappa(0)=0$. For interfering BSs, $\kappa(n)=$
\begin{equation}
\label{eq:kappa2} \frac{3.5^{3.5}}{n!}\frac{\Gamma(n+3.5)}{\Gamma(3.5)}\rho^{n} \left(3.5+\rho\right)^{-n-3.5},
\end{equation}
for $n\geq 0$.
\fi

Assuming $N_y$ to be independently and identically distributed, we model the PMF of $U_y$ by 
%\begin{equation}
%\label{eq:pn}
%\mathbb{P}\left(U_\x = n\right) =  \begin{cases}
%\kappa(n) & \text{ if }0\leq  n\leq \NRF-1\\
%1-\sum_{i=1}^{\NRF-1}\kappa(n) & \text{if }n = \NRF\\
%0 & \text{otherwise,}
%\end{cases}
%\end{equation}
\ifSC
\begin{equation}
\label{eq:pn}
\mathbb{P}\left(U_y = n\right) =  \mathds{1}_{\{0\leq  n\leq \NRF-1\}}
\kappa(n) + \left(1-\sum_{i=1}^{\NRF-1}\kappa(n)\right) \mathds{1}_{\{n=\NRF\}}
\end{equation}
\else
\begin{multline}
\label{eq:pn}
\mathbb{P}\left(U_y = n\right) =  \mathds{1}_{\{0\leq  n\leq \NRF-1\}}
\kappa(n) \\+ \left(1-\sum_{i=1}^{\NRF-1}\kappa(n)\right) \mathds{1}_{\{n=\NRF\}}
\end{multline}
\fi

We first model the point process $\PPP$ to be superposition of the point processes $\Phi_\mathrm{L}$ and $\Phi_\mathrm{N}$ with intensities $\dnsty \plos \mathds{1}_{\{||x||\leq\mathrm{D}\}}$ and $\dnsty (1-\plos) \mathds{1}_{\{||x||\leq\mathrm{D}\}}+\dnsty\mathds{1}_{\{||x||> \mathrm{D}\}}$, respectively. These two point processes correspond to LOS and NLOS BSs. 
%Given that there is atleast one LOS BS, we define $x^*_{\mathrm{L}} = \underset{y\in\Phi_\mathrm{L}}{\arg\min}\,{L(y,0)}$. Probability that there is atleast one LOS BS is given by $\mathrm{B}_\mathrm{L} = 1-\exp\left(-\dnsty \plos\pi \mathrm{D}^2\right)$. Similarly, given that there is \\
% one NLOS BS, we define $x^*_{\mathrm{L}} = \underset{y\in\Phi_\mathrm{L}}{\arg\min}\,{L(y,0)}$. Note that the probability that there is at least one NLOS BS can be found to be $\mathrm{B}_\mathrm{N}=1$. 
%Note that if $\x\in\PPP$ is the {\em tagged} BS, then $x = \underset{x\in\{\x^*_\mathrm{L},\x^*_\mathrm{N}\}}{\arg\min}(L(\x^*_\mathrm{L},0),L(\x^*_\mathrm{N},0))$.
The corresponding propagation processes \cite{BlaKarKee12} are given as
$\mathcal{N}_\mathrm{L} = \{{||y||^{\alpha_\mathrm{L}}}/{S_{y,\mathrm{L}}} : y\in\Phi_\mathrm{L}\}$, and 
$\mathcal{N}_\mathrm{N} = \{{||y||^{\alpha_\mathrm{N}}}/{S_{y,\mathrm{N}}} : y\in\Phi_\mathrm{N}\}$.
%The notion of propagation process was proposed in \cite{BlaKarKee12} for the analysis of cellular networks. Using propagation process as a tool to analyze $\SINR$ coverage helps incorporate log-normal shadow fading in the channel and in user association in a tractable way. 

\begin{lem}
\label{lem:mdnstyl}
$\mathcal{N}_\mathrm{L}$ is a non-homogeneous PPP with intensity $\Lambda_\mathrm{L}([0,t))= \dnsty \mathrm{M}_\mathrm{L}(t)$,
where 
\ifSC
\begin{equation}
\mathrm{M}_\mathrm{L}(t) = \pi \plos\left[\dmax^2\Q\left(\Upsilon_{\mathrm{L}}(t)\right)  +  t^{\frac{2}{{\ple_{\mathrm{L}}}}}\exp\left(\frac{2\sigma_{\mathrm{L}}^2}{\ple_{\mathrm{L}}^2}+ \frac{2\meen}{\ple_{\mathrm{L}}} \right) \Q\left({\frac{ 2\sigma_{\mathrm{L}}^2}{\ple_{\mathrm{L}}\sigma_\mathrm{L}} -\Upsilon_{\mathrm{L}}(t)}\right)
\right]. 
\end{equation}
\else
$\mathrm{M}_\mathrm{L}(t)$ is given by (\ref{eq:MLt}). 
\begin{figure*}
\begin{equation}
\label{eq:MLt}
\mathrm{M}_\mathrm{L}(t) = \pi \plos\left[\dmax^2\Q\left(\Upsilon_{\mathrm{L}}(t)\right)  +  t^{\frac{2}{{\ple_{\mathrm{L}}}}}\exp\left(\frac{2\sigma_{\mathrm{L}}^2}{\ple_{\mathrm{L}}^2}+ \frac{2\meen}{\ple_{\mathrm{L}}} \right) \Q\left({\frac{ 2\sigma_{\mathrm{L}}^2}{\ple_{\mathrm{L}}\sigma_\mathrm{L}} -\Upsilon_{\mathrm{L}}(t)}\right)
\right]. 
\end{equation}
\end{figure*}
\fi
Here, $m = -0.1\fsldb\ln10$, $\sigma_{\mathrm{L}}=0.1 \sdpl_\mathrm{L}\ln10$, $\Upsilon_{j}(t) = \frac{\ln(\frac{\dmax^{\ple_{j}}}{t})-\meen}{\sigma_{j}}$ for $j \in \{\mathrm{L},\mathrm{N}\}$and $\Q(.)$ is the Q-function (Standard  Gaussian CCDF).
\end{lem}
\begin{proof}
%See Appendix~\ref{app:propprocess}. Note that this is a special case of the result in Appendix~\ref{app:propprocess} of \cite{SinJSAC14}.
Special case of Appendix~A of \cite{SinJSAC14} and is therefore skipped for brevity.
\end{proof}

\begin{lem}
\label{lem:mdnstyn}
$\mathcal{N}_\mathrm{N}$ is a non-homogeneous PPP with intensity
$
\Lambda_\mathrm{N}([0,t)) = \dnsty \mathrm{M}_\mathrm{N}(t)
$, where 
\ifSC
\begin{equation}
\mathrm{M}_\mathrm{N}(t) = -\pi \plos \dmax^2\Q\left(\Upsilon_\mathrm{N}(t)\right)  + \pi t^\frac{2}{{\ple_{\mathrm{N}}}}\exp\left(\frac{2\sigma_{\mathrm{N}}^2}{\ple_{\mathrm{N}}^2}+ \frac{2 m}{\ple_{\mathrm{N}}} \right) \left[1-\plos \Q\left( \frac{2\sigma_{\mathrm{N}}^2}{\ple_{\mathrm{N}}\sigma_{\mathrm{N}}}-\Upsilon_\mathrm{N}(t)\right)\right].
\end{equation}
\else
$\mathrm{M}_\mathrm{N}(t)$ is given by (\ref{eq:MNt}). 
\begin{figure*}
\begin{equation}
\label{eq:MNt}
\mathrm{M}_\mathrm{N}(t) = -\pi \plos \dmax^2\Q\left(\Upsilon_\mathrm{N}(t)\right)  + \pi t^\frac{2}{{\ple_{\mathrm{N}}}}\exp\left(\frac{2\sigma_{\mathrm{N}}^2}{\ple_{\mathrm{N}}^2}+ \frac{2 m}{\ple_{\mathrm{N}}} \right) \left[1-\plos \Q\left( \frac{2\sigma_{\mathrm{N}}^2}{\ple_{\mathrm{N}}\sigma_{\mathrm{N}}}-\Upsilon_\mathrm{N}(t)\right)\right].
\end{equation}
\end{figure*}
\fi
\end{lem}
\begin{proof}
Proceeds very similarly to Lemma~\ref{lem:mdnstyl} and thus is omitted. 
\end{proof}
Note that here $\Lambda_\mathrm{L}([0,\infty)) = \dnsty \pi \plos\mathrm{D}^2$. The probability that there is no point in the interval $[0,\infty)$ is equal to $ \exp\left(-\dnsty \pi \plos\mathrm{D}^2\right)$. This is exactly the probability that there is no point in $\Phi_\mathrm{L}$.  Let us call the probability that there is at least one point in $\mathcal{N}_\mathrm{L}$ to be $\mathrm{B}_\mathrm{L}$ .The event that number of points in $\Phi_\mathrm{N}$ is zero is empty and thus, $\mathrm{B}_\mathrm{N} = 1$.

\begin{cor}
Let $\mathcal{N}$ be the point process of propagation losses corresponding to $\PPP$. This point process is a PPP with intensity $\Lambda((0,t]) = \dnsty (\mathrm{M}_\mathrm{L}(t)+\mathrm{M}_\mathrm{N}(t))= \dnsty \mathrm{M}(t)$. 
\end{cor}
\begin{proof}
Follows directly from the Superposition property of PPPs\cite[Proposition 1.3.3]{BacBook09}.
\end{proof}
\begin{figure*}
\begin{multline}
\label{eq:MLdash}
\mathrm{M}^{'}_\mathrm{L}(t) =  \pi\plos\Bigg\{
\exp\left(\frac{2\sigma_{\mathrm{L}}^2}{\ple_{\mathrm{L}}^2}+ \frac{2m}{\ple_{\mathrm{L}}} \right) t^{\frac{2}{\ple_{\mathrm{L}}}-1} \Bigg[\frac{2}{\ple_\mathrm{L}}\Q\left(\frac{2\sigma_{\mathrm{L}}^2}{\ple_{\mathrm{L}} \sigma_\mathrm{L}} -\Upsilon_\mathrm{L}(t)\right)\Bigg.\\ \left.-\frac{1}{\sqrt{2\pi\sigma_\mathrm{L}^2}}\exp\left(-\left(\frac{\sqrt{2}\sigma_{\mathrm{L}}}{\ple_{\mathrm{L}}} - \frac{\Upsilon_{\mathrm{L}}(t)}{\sqrt{2}}\right)^2\right)\right]
+ \frac{\dmax^2}{\sqrt{2\pi} t\sigma_\mathrm{L}}\exp\left(-\frac{\Upsilon^2_\mathrm{L}(t)}{2}\right)\Bigg\}.
\end{multline}
\end{figure*}
\begin{figure*}
\begin{multline}
\label{eq:MNdash}
\mathrm{M}^{'}_\mathrm{N}(t) = \pi\plos\Bigg\{
\exp\left(\frac{2\sigma_{\mathrm{N}}^2}{\ple_{\mathrm{N}}^2}+ \frac{2m}{\ple_{\mathrm{N}}} \right) t^{\frac{2}{\ple_{\mathrm{N}}}-1}
\Bigg[\frac{2}{\plos\ple_{\mathrm{N}}}-\frac{2}{\ple_\mathrm{N}}\Q\left({\frac{2\sigma_{\mathrm{N}}^2}{\ple_{\mathrm{N}}\sigma_\mathrm{N}} -\Upsilon_\mathrm{N}(t)}\right)\Bigg.\\ \left.+\frac{1}{\sqrt{2\pi\sigma_\mathrm{N}^2}}\exp\left(-\left({\frac{\sqrt{2}\sigma_{\mathrm{N}}}{\ple_{\mathrm{N}}} - \frac{\Upsilon_\mathrm{N}(t)}{\sqrt{2}}}\right)^2\right)\right]
-\frac{\dmax^2}{\sqrt{2\pi} t\sigma_\mathrm{N}}\exp\left(-\frac{\Upsilon^2_\mathrm{N}(t)}{2}\right)\Bigg\}.
\end{multline}
\end{figure*}
\begin{lem}
\label{lem:mdnsty}
Given that $\mathcal{N}_\mathrm{L}$ and $\mathcal{N}_\mathrm{N}$ are not empty, the probability density function (PDF) of the distance to the point nearest to origin in these point processes is given by 
$f_\mathrm{L}(t)  = \dnsty\exp\left(-\dnsty \mathrm{M}_\mathrm{L}(t)\right) \mathrm{M}^{'}_\mathrm{L}(t)/\mathrm{B}_\mathrm{L}$ and $f_\mathrm{N}(t)  = \dnsty \exp\left(-\dnsty \mathrm{M}_\mathrm{N}(t)\right)
\mathrm{M}^{'}_\mathrm{N}(t)/\mathrm{B}_\mathrm{N}$,
where $\mathrm{M}^{'}_\mathrm{L}(t)$ and $\mathrm{M}^{'}_\mathrm{N}(t)$ are given in (\ref{eq:MLdash}) and (\ref{eq:MNdash}), respectively.
\end{lem}
\begin{proof}
If $l^*$ is the point nearest to origin in the point process $\mathcal{N}_L$, 
\ifSC
\begin{align*}
&\mathbb{P}\left(l^*>t\big|\mathcal{N}_L([0,\infty))>0\right)
= \mathbb{P}\left(\mathcal{N}_L([0,t))=0\big|\mathcal{N}_L([0,\infty))>0\right)\\
&=\mathbb{P}\left(\mathcal{N}_L([0,t))=0\cap\mathcal{N}_L([0,\infty))>0\right)/\mathbb{P}\left(\mathcal{N}_L([0,\infty))>0\right)
=\mathbb{P}\left(\mathcal{N}_L([0,t))=0\cap\mathcal{N}_L([t,\infty))>0\right)/\mathrm{B}_\mathrm{L}
\end{align*}
\begin{align*}
&=\mathbb{P}\left(\mathcal{N}_L([0,t))=0\right)\pr \left(\mathcal{N}_L([t,\infty))>0\right)/\mathrm{B}_\mathrm{L}
=\exp\left(-\Lambda_\mathrm{L}([0,t)]\right)\left(1-\exp\left(-\Lambda_\mathrm{L}([t,\infty)]\right)\right)/\mathrm{B}_\mathrm{L}\\
&=\left(\exp\left(-\Lambda_\mathrm{L}([0,t)]\right)-\exp\left(-\Lambda_\mathrm{L}([0,\infty)]\right)\right)/\mathrm{B}_\mathrm{L}.
\end{align*}
\else
\begin{align*}
&\mathbb{P}\left(l^*>t\big|\mathcal{N}_L([0,\infty))>0\right)  \\&
= \mathbb{P}\left(\mathcal{N}_L([0,t))=0\big|\mathcal{N}_L([0,\infty))>0\right)\\
 &=\frac{\mathbb{P}\left(\mathcal{N}_L([0,t))=0\cap\mathcal{N}_L([0,\infty))>0\right)}{\mathbb{P}\left(\mathcal{N}_L([0,\infty))>0\right)}
\\&
=\mathbb{P}\left(\mathcal{N}_L([0,t))=0\cap\mathcal{N}_L([t,\infty))>0\right)/\mathrm{B}_\mathrm{L} %\\
 \end{align*}
 \begin{align*}
&=\mathbb{P}\left(\mathcal{N}_L([0,t))=0\right)\mathbb{P} \left(\mathcal{N}_L([t,\infty))>0\right)/\mathrm{B}_\mathrm{L}
 \\&
=\exp\left(-\Lambda_\mathrm{L}([0,t)]\right)\left(1-\exp\left(-\Lambda_\mathrm{L}([t,\infty)]\right)\right)/\mathrm{B}_\mathrm{L}\\
&=\left(\exp\left(-\Lambda_\mathrm{L}([0,t)]\right)-\exp\left(-\Lambda_\mathrm{L}([0,\infty)]\right)\right)/\mathrm{B}_\mathrm{L}.
\end{align*}
\fi

Thus, taking the negative derivative of the above expression we get the PDF $f_\mathrm{L}(t) $. Similarly, we can derive the PDF for the NLOS case.
\end{proof}

\begin{thm}
\label{thm:sinr}
The $\SNR$ coverage of a typical user in the network is given by
\begin{equation}
\label{eq:sinrthmnoiselimit}
\mathcal{\overline{S}}(\tau) \triangleq \pr(\SNR_{x,0} > \tau)= \cexpect{\mathcal{S}(\tau,U_\x)}{U_\x}, \text{ where }
\end{equation}
\begin{multline*}
\mathcal{S}(\tau,\mathrm{U}) \approx \sum_{j\in\{\mathrm{L,N}\}} \mathrm{B}_j \zeta(\eta_j,\mathrm{U}) \sum_{n=1}^{\eta_j}(-1)^{n+1}{\eta_j \choose n}
\ifSC
\else \\\times
\fi
\int\limits_{0}^{\infty}\exp\left(-\frac{\eta_j\tau n \mathrm{U} l \sigma^2_n}{\mathrm{G}}-\dnsty \mathrm{M}_{\overline{j}}(l)
\right) f_{j}(l)\mathrm{d}l,
\end{multline*} 
where $\mathrm{G} = \power{} \NBS \NMS$, $\overline{j} = \mathrm{L}$ if $j = \mathrm{N}$ and vice versa. 
The terms $\zeta(.)$, $\mathrm{M}_j(.)$ and $f_{(.)}$ are derived in Proposition~\ref{prop:1}, Lemma~\ref{lem:mdnstyl}, Lemma~\ref{lem:mdnstyn} and Lemma~\ref{lem:mdnsty}.
\end{thm}
\begin{proof}
See Appendix~\ref{app:sinr}
\end{proof}
\begin{cor}
Assuming that user density is much larger than BS density, the $\SNR$ coverage can be approximated by $\mathcal{S}(\tau,\mathrm{U}_\mathrm{M})$.
\end{cor}
%\begin{cor}
%The $\SNR$ coverage of a typical user, if it has a single omni-directional antenna is as given in Theorem~\ref{thm:sinr} but with %$\zeta_{(.), \mathrm{UE}}(.)=1$ and $\NMS=1$.
%\end{cor}

\begin{thm}
\label{thm:rate}
In a noise-limited network, the per user rate distribution (or rate coverage) of a {\em typical} user at origin served by a BS at $\x$ is given by
\ifSC
\begin{align*}
\mathcal{R}(\tau_r)\triangleq \mathbb{P}\left(R_{\x,0}>\tau_r\right)
= \sum_{n\geq 1}\kappa(n)\mathcal{S}\left(2^{\frac{\tau_r n}{\omega_{\mathrm{MU}} \res \min\left(n,\NRF\right)}}-1,\min\left(n,\NRF\right)\right),
 \end{align*}
\else 
\begin{align*}
\mathcal{R}(\tau_r)&\triangleq \mathbb{P}\left(R_{\x,0}>\tau_r\right) \\
&= \sum_{n\geq 1}\kappa(n)\mathcal{S}\left(2^{\frac{\tau_r n}{\omega_{\mathrm{MU}} \res \min\left(n,\NRF\right)}}-1,\min\left(n,\NRF\right)\right),
\end{align*}
\fi
where $\mathcal{S}(.)$ was defined in Theorem~\ref{thm:sinr} and $\kappa(n)$ is given in (\ref{eq:kappa}).
\end{thm}
\begin{proof}
Follows by re-arranging (\ref{eq:peruserrate}) and using $\SNR=\SINR$. 
\end{proof}
Although the above expression is an infinite summation, as verified earlier in \cite{SinDhiAnd13,SinJSAC14}, it can be accurately represented as a finite summation. For the results in this work, considering the first $\lfloor 12\lambda_\mathrm{UE}/\lambda_\mathrm{BS}\rfloor $ terms is sufficient. The following definition will be useful when comparing the rate coverage of MU-MIMO with SM and SU-BF.
\begin{definition}
\label{def:releffi}
The minimum allowable efficiency of scheme A such that it is guaranteed to outperform scheme B in terms of per user rate for $p$ percentile users (that is users with rate coverage $p$), is given by 
$
\mathcal{O}_{\mathrm{A},\mathrm{B}}(p) = \frac{\mathcal{R}^{-1}_\mathrm{B}(p)}{\mathcal{R}^{-1}_\mathrm{A}(p)},
$
where $\mathcal{R}^{-1}_\mathrm{A}$ and $\mathcal{R}^{-1}_\mathrm{B}$ are inverse of the rate coverage at $p$ (that is rate thresholds $\tau$ corresponding to $\mathcal{R}(\tau) = p$) for schemes A and B after setting $\omega_\mathrm{A} = \omega_\mathrm{B} = 1$, where $\omega_{(.)}$ are the efficiency factors for the respective MIMO techniques as defined in (\ref{eq:peruserrate}). The per user rate of A cannot stochastically dominate that of B, unless the efficiency of A is at least $\min_p \mathcal{O}_{\mathrm{A},\mathrm{B}}(p)$.
\end{definition}

Note that MU-MIMO implementations with different $\NRF$ are considered as separate MIMO schemes in the above definition since they have different efficiency factors. It is clear that  $\frac{\mathcal{O}_{\mathrm{A},\mathrm{B}}(p)}{\omega_\mathrm{A}}$ is an upper bound on the ratio $\mathcal{R}^{-1}_\mathrm{B}(p)$ and $\mathcal{R}^{-1}_\mathrm{A}(p)$ for non-unity efficiency for scheme A and $\omega_\mathrm{B}=1$, since $\omega_\mathrm{A}$ is the minimum efficiency over all BSs in the network. Note that from  (\ref{eq:peruserrate}), setting $\omega_\mathrm{B}=1$ gives upper bound on rate for scheme B. For scheme A to outperform scheme B for $p$ percentile users, we need  $\frac{\mathcal{O}_{\mathrm{A},\mathrm{B}}(p)}{\omega_\mathrm{A}}\leq 1$ with equality giving minimum allowable $\omega_\mathrm{A}$. If the network is such that all BSs see the same overheads $\omega_\mathrm{A}$ or $\omega_\mathrm{B}$, then $\mathcal{O}_{\mathrm{A},\mathrm{B}}(p)$ is the minimum allowable relative efficiency (that is $\omega_\mathrm{A}/\omega_\mathrm{B}$) of scheme A over scheme B. This gives tighter estimates for allowable $\omega_\mathrm{A}$ especially when comparing MU-MIMO for different $\NRF$ or SM for different number of streams. 

\subsubsection{Rate Distribution in an Interference-limited Network}
Until now, our analysis focused on noise-limited mmWave cellular networks. In this section, we will discuss how to model interference in these networks.

From (\ref{eq:sinr1}), the OCI power at user $\us$ served by a BS at $\x$ is modelled as 
\ifSC
\begin{equation*}
I_\us =\power{} \sum\limits_{\substack{y\in\PPP\\y\neq \x}} \sum\limits_{w\in\mathcal{U}_y} \frac{||\overline{\bf h}^*_{y,\us} \bbcolumn{w}{y}||^2}{U_y}= \power{} \sum\limits_{\substack{y\in\PPP\\y\neq \x}} \sum\limits_{w\in\mathcal{U}_y} \frac{||{\bf w}^{*}_\us {\bf H}_{y,\us} {\mathrm{\bf F}^\mathrm{RF}_y} \bbcolumn{w}{y}||^2}{U_y} .
\end{equation*}
\else
\begin{align*}
I_\us &=\power{} \sum\limits_{\substack{y\in\PPP\\y\neq \x}} \sum\limits_{w\in\mathcal{U}_y} \frac{||\overline{\bf h}^*_{y,\us} \bbcolumn{w}{y}||^2}{U_y}\\&= \power{} \sum\limits_{\substack{y\in\PPP\\y\neq \x}} \sum\limits_{w\in\mathcal{U}_y} \frac{||{\bf w}^{*}_\us {\bf H}_{y,\us} {\mathrm{\bf F}^\mathrm{RF}_y} \bbcolumn{w}{y}||^2}{U_y} .
\end{align*}
\fi
Here, ${\bf w}_\us = \aMS(\phi_{\x,\us})$, ${\bf H}_{y,\us} = \sqrt{\frac{\NBS\NMS}{L(y,\us)\eta_{y,\us}}}\sum_{i=1}^{\eta_{y,\us}}\gamma_{i,y,\us} \aMS(\AOA{i}{y}{\us}) \mathrm{\bf a}^{*}_\mathrm{BS}(\AOD{i}{y}{\us})$,  ${\mathrm{\bf F}^\mathrm{RF}_y}$ has columns equal to $\aBS(\theta_{y,w})$ for all $w\in\mathcal{U}_y$,  and $\bbcolumn{w}{y}$ is designed so as to cancel the multiuser interference of the BS at $y$. All the AOAs and AODs in the above expression are independent of each other. Leveraging the virtual channel approximation for large number of antennas at the BS and UE, interference due to the link between BS at $y$ and user at $w$ on the user $u$ is non-zero if and only if $\phi_{\x,\us}$ is equal to at least one of the AOA of ${\bf H}_{y,\us}$ and $\theta_{y,w}$ equals the corresponding AOD. Note that since multiuser interference was cancelled by the ZF precoder, the virtual approximation with an ON/OFF model for inner product of two beam steering vectors gave us a tractable and accurate tool to study $\SNR$ distribution in the previous section. However, this model may not be accurate when OCI is incorporated. 

The virtual channel approximation quantized the angular space into $\mathrm{N}$ sectors, where $\mathrm{N}$ is the number of antennas. If two angles lie on either sides of a sector boundary, the inner product of beam steering vectors is zero, which can be a main cause of underestimated interference. In order to alleviate this problem, we introduce the notion of sidelobe gain which was also used in \cite{BaiHea14, SinJSAC14}. We still assume that the virtual angle space is quantized into $\mathrm{N}$ sectors with the angle bisector being a representative of each sector, but the inner product between two beam steering vectors is defined as:
\begin{equation}
\label{eq:innerprodapprox}
\aBS^{*}(\theta_1) \aBS(\theta_2)\triangleq \begin{cases}
1 & \text{ if }\theta_1= \theta_2\\
\rho_\mathrm{BS} & \text{ otherwise,}
\end{cases}
\end{equation}
where $\rho_{\mathrm{BS}}<1$ introduces a sidelobe gain into the model. Similarly, we model the inner product for beam steering vectors at UEs with parameter $\rho_\mathrm{UE}$. Note that setting $\rho_{\mathrm{BS}}=\rho_{\mathrm{UE}}=0$ reverts back to the virtual channel approximation. 

To characterize the interference distribution, we neglect the effect of ZF on interfering links and dependence in $p_{\mathrm{MU}}$ and $I_u$ through $\{{\bf w}_u\}$ for tractability and show that for a fairly large number of antennas this is a reasonable approximation. 
First we deal with the $\eta_\mathrm{L} = \eta_\mathrm{N} = 1$ case.
\begin{prop}
\label{prop:2}
Assuming the inner product of any two beam steering vectors at BS or UE follow the law given by (\ref{eq:innerprodapprox}),  $\eta_\mathrm{L} = \eta_\mathrm{N} = 1$ and propagation loss on the service link is $l$, the OCI power at the typical user can be modelled as 
$
I_{0} = \sum\limits_{\substack{y\in\PPP,y\neq\x} } \mathrm{G}|\gamma_{y,0}|^2 L(y,0)^{-1} \chi_y/U_y,
$
where  $\gamma_{y,0}$ is complex normal random variable with unit variance and zero mean, $U_y$ are i.i.d random variables with distribution given in (\ref{eq:pn}) and $\chi_y$ is defined as 
\begin{multline*}
\iftoggle{SC}{}{\hspace{-0.5cm}}\chi_y = \begin{cases}
k+(U_y-k)\rho^2_\mathrm{BS} & {\hspace{-0.5cm}}\text{w.p. } (\sum_{i=1}^{\NMS}q^2_{\mathrm{UE},i}) \iftoggle{SC}{{{U_y}\choose{k}}}{\times\\&{\hspace{-1.5cm}}} {{U_y}\choose{k}}\sum_{j=1}^{\NBS} q^{k+1}_{\mathrm{BS},j} (1-q_{\mathrm{BS},j})^{U_y-k} \\
\rho^2_\mathrm{UE} (k+(U_t-k)\rho^2_\mathrm{BS})  & {\hspace{-0.25cm}}\text{w.p. } (1-\sum_{i=1}^{\NMS}q^2_{\mathrm{UE},i}) \iftoggle{SC}{{{U_y}\choose{k}}}{\times\\&{\hspace{-1.5cm}}}{{U_y}\choose{k}} \sum_{j=1}^{\NBS} q^{k+1}_{\mathrm{BS},j} (1-q_{\mathrm{BS},j})^{U_y-k}, \\
\end{cases}
\end{multline*}
for $k=1,2,\ldots,U_y$.
\end{prop}
\begin{proof}
For single path channel, the out-of-cell interference is given by 
\begin{align*}
I_{0} &= \sum_{y\in\PPP,y\neq x}\frac{\mathrm{G}|\gamma_y|^2 L(y,u)^{-1}}{{U_y}} \iftoggle{SC}{}{\times\\&\hspace{0.5cm}} \sum_{w\in\mathcal{U}_y} || \aMS^*(\phi_{x,u} )\aMS(\phi_{y,u})\aBS^*(\theta_{y,u})\aBS(\theta_{y,w})||^2.
\end{align*}
Now using the inner product rule in (\ref{eq:innerprodapprox}) and the fact that all the virtual angles in the above equation are independent and distributed according to Lemma~\ref{lem:angle}, the proposition can be proved.
\end{proof}
%In reality, $U_y$, all the AOAs and AODs are not independent random variables. However, this assumption enables tractable analysis. Numerical validation in \cite{SinDhiAnd13} for conventional cellular networks and that in \cite{SinJSAC14} for mmWave networks provides confidence for the assumption on independent $U_y$. 

\begin{lem}
\label{lem:int}
The Laplace functional of the interference power in Proposition~\ref{prop:2} conditioned on path loss to the {\em typical} user at origin from serving BS is $L(\x,\us) = l$, is given by 
\begin{align*}
\iftoggle{SC}{}{&}L_{I_0,l}(s)  \iftoggle{SC}{&}{}\triangleq \expect{\exp\left(-s I_{0}\right)|L(\x,0) = l} \\
&=
\mathrm{exp}\left(-\dnsty\sum_{n=0}^{\NRF}\tilde{p}(n)\sum_{k=0}^{n}{n\choose k} \sum_{i=1}^{\NBS} q^{k+1}_{\mathrm{BS},i}(1-q_{\mathrm{BS},i})^{n-k} \right.\\
&\iftoggle{SC}{}{\hspace{2cm}}
\left.\left\lbrace\left(\sum_{i=1}^{\NMS}q^2_{\mathrm{UE},i}\right) \int_{t\geq l}\frac{\mathrm{M}^{'}(t) \mathrm{d}t}{1+\frac{tn}{s\mathrm{G} (k+(n-k)\rho^2_\mathrm{BS})}}
\iftoggle{SC}{}{\right.\right.\\&\hspace{0.5cm}\left.\left.} +\left(1-\sum_{i=1}^{\NMS}q^2_{\mathrm{UE},i}\right) \int_{t\geq l}\frac{\mathrm{M}^{'}(t) \mathrm{d}t}{1+\frac{tn}{s\mathrm{G} \rho^2_\mathrm{UE}(k+(n-k)\rho^2_\mathrm{BS})}}\right\rbrace\right).
\end{align*}
where $\tilde{p}(.)$ is the distribution of $U_y$ for interfering BSs given in (\ref{eq:pn}).
\end{lem}
\begin{proof}
Appendix~\ref{app:int}.
\end{proof}
\begin{thm}
\label{thm:sinr2}
The $\SINR$ coverage of the typical user is given by (\ref{eq:sinrthmnoiselimit}) with an extra term $L_{I_0,l}\left(\frac{\eta_{j} \tau n \mathrm{U} l}{\mathrm{G}}\right)$ inside the integral over $\mathrm{d}l$.
%\begin{equation}
%\mathcal{\overline{S}}(\tau) \triangleq \pr(\SINR_{x,0} > \tau)= \cexpect{\mathcal{S}(\tau,U_\x)}{U_\x}, \text{ where}
%\end{equation}
%\begin{multline}
%\mathcal{S}(\tau,\mathrm{U}) \approx
 %\sum_{j\in\{\mathrm{L,N}\}}\mathrm{B}_j (\zeta_{\mathrm{UE},j}(\eta_j) +\zeta_{\mathrm{BS},j}(\eta_j,\mathrm{U}) -\zeta_{\mathrm{UE},j}(\eta_j) %\zeta_{\mathrm{BS},j}(\eta_j,\mathrm{U}) )\times\\
%\sum_{n=1}^{\eta_j}(-1)^{n+1}{\eta_j \choose n}\int_{0}^{\infty}\exp\left(-\frac{\eta_j\tau n \mathrm{U} l \sigma^2_n}{\mathrm{G}}-\dnsty %\mathrm{M}_{\overline{j}}(l)
%\right) L_{I_0,l}\left(\frac{\eta_{j} \tau n \mathrm{U} l}{\mathrm{G}}\right) f_{j}(l)\mathrm{d}l.
%\end{multline} 
%where $\mathrm{G} = \power{} \NBS \NMS$, $\overline{j} = \mathrm{L}$ if $j = \mathrm{N}$ and vice versa. Here, $\zeta_{\mathrm{UE},j}(\eta) = %\sum_{j=1}^{\NMS}q_{\mathrm{UE},j}(1-q_{\mathrm{UE},j})^{\eta-1}$, 
%\begin{equation}
%\zeta_{\mathrm{BS},j}(\eta,\mathrm{U}) = %\sum_{j=1}^{\NBS}q_{\mathrm{BS},j}(1-q_{\mathrm{BS},j})^{\eta-1}\left[\plos\left(1-q_{\mathrm{BS},j}\right)^{\eta_\mathrm{L}}+
%(1-\plos)\left(1-q_{\mathrm{BS},j}\right)^{\eta_\mathrm{N}}
%\right]^{\mathrm{U}-1}.
%\end{equation}
%The terms $\mathrm{M}_j(.)$ and $f_j(.)$ are derived in Lemma~\ref{lem:mdnstyl}, Lemma~\ref{lem:mdnstyn} and Lemma~\ref{lem:mdnsty}. 
\end{thm}
\begin{proof}
Exactly on same lines as Theorem~\ref{thm:sinr}. The Laplace functional $L_{I_0,l}(.)$ has been derived in Lemma~\ref{lem:int} for single path channel. Upper and lower bounds on $L_{I_0,l}(.)$ for a general number of paths can be found in Appendix \ref{app:int2}.%, but here we use the modified virtual channel approximation to model the out of cell interference. The received signal power is the same as Proposition~\ref{prop:1}. 
\end{proof}
From this expression of $\SINR$ coverage, the rate coverage can be found similar to Theorem~\ref{thm:rate}. 
We will validate these analytical results in Section~\ref{sec:results}. In the next section, we will take a brief look into the coverage and rate for SM enabled mmWave cellular networks. Before that though, we provide a brief discussion on how to choose $\rho_\mathrm{UE}$ and $\rho_\mathrm{BS}$. Recall that $\NBS\rho^2_{\mathrm{BS}}$ and $\NMS\rho^2_{\mathrm{UE}}$ are the sidelobe gains for beam pattern at BSs and UEs, respectively. An obvious question is whether these parameters depend on the number of antennas and if yes, how should their dependence be modelled? 
%\begin{figure}
%\centering
%\includegraphics[width=0.5\columnwidth, trim = 0.9cm 0.3cm 1cm 0.6cm]{Fig9}
%\caption{Interference power distribution for varying number of antennas in a single path channel.}
%\label{fig:intsim}
%\end{figure}

%To address this question, we first plot the interference power for different number of antennas at the BS in Fig.~\ref{fig:intsim}. 
%As can be seen from the figure, the interference power reduces with increasing $\NBS$. This is in line with 
If $\rho_{(.)}$ were to be a constant, the sidelobe gain will also scale up with an increasing number of antennas. This will violate Lemma~\ref{lem:precodcombin}. Since virtual channel approximation asymptotically tracks physical channel model, $\rho_\mathrm{BS}$ and $\rho_\mathrm{UE}$ should decrease and eventually vanish with increasing $\NBS$ and $\NMS$, respectively. For a uniform linear array with $\mathrm{N}$ antennas, the ratio of the gain of the $i^\text{th}$ sidelobe to the main lobe is equal to 
$|\frac{\sin\left(0.5\pi (2i+1)\right)}{\mathrm{N}\sin\left(0.5\pi (2i+1)/\mathrm{N}\right)}|^2$\cite{balanis}, for $i=1, 2, \ldots, \lfloor \frac{\mathrm{N}}{2}-1\rfloor$. For $i\ll\mathrm{N}$, this ratio is independent of $\mathrm{N}$ using the small angle approximation $\sin \theta\approx \theta$. For $i$ on the order of $\mathrm{N}$, this ratio decreases approximately as square of $\mathrm{N}$. The regime in which the ratio is independent of $\mathrm{N}$ has about fixed beamwidth, which corresponds to the beamwidth in which the small angle approximation of $\sin \theta\approx \theta$ is accurate with $p$ percent relative error. For $p=1$, $\theta\approx 0.244$ radians. Since the majority of the angular space corresponds to the regime in which the above ratio varies inversely with the square of $\mathrm{N}$, we model $\rho_{(.)}$ to linearly decrease with $\mathrm{N}$. We choose $\rho_\mathrm{BS} = 1/(\sin(0.244) \NBS)$ and $\rho_\mathrm{UE} = 1/(\sin(0.244) \NMS)$, however in future it is desirable to re-investigate the scaling factor to get a better fit. 

\section{Single User Spatial Multiplexing in mmWave Cellular Networks}
\label{sec:SM}
For spatial multiplexing (SM), we consider a scenario in which every BS transmits more than one stream of data to a single user per resource block. Thus, $\mathrm{N}^{\mathrm{UE}}_s=\mathrm{N}^{\mathrm{BS}}_s=\mathrm{N}_s$, where $\mathrm{N}_s$ is the multiplexing gain.  In this section, we will focus mainly on the multipath diversity approach for SM\cite{Aya14,Sun14} and not on the polarization approach\cite{Sun14,Ghosh14}. 
\subsection{Spatial Multiplexing: UHF versus mmWave}
We begin with a brief recap of the theoretically optimal implementation of closed-loop SM in conventional cellular networks, which motivates the main challenges in precoding/combining for SM in mmWave networks. Under unitary power constraint, given the singular value decomposition of the channel matrix ${\bf H} = {\bf U}{\bf \Sigma} {\bf V}^*$, the transmitter pre-multiplies the input symbols with matrix ${\bf V}$ and the receiver combines the received signal on all its antennas with matrix ${\bf U}^*$, to effectively achieve $\mathrm{N}_s$ parallel channels, where $\mathrm{N}_s$ is the multiplexing gain. Since the channel matrix is either full row rank or full column rank with high probability for sub $6$ GHz frequency bands, $\mathrm{N}_s = \min\{\NBS,\NMS\}$. 

At mmWave frequencies, however, the first challenge is that it is not practically feasible to implement a fully digital precoder and combiner. Using the popular hybrid beamforming approach for mmWave networks\cite{AlkMo14}, the precoder is of the form $\rfprecoder\bbprecoder$, wherein $\rfprecoder$ is generally implemented using phase shifters and thus has constant magnitude entries. Similarly, we have a constraint for the combiner. Another challenge for implementing SM at mmWave is that the channel is  sparse \cite{Sun14,SamRap14} and thus obtaining multiplexing gain on the order of number of antennas is nearly impossible. 

We now look at a typical implementation of SM using the hybrid beamforming architecture in Fig.~\ref{fig:hybridarchitecture}. Assuming perfect channel estimation, and using the system model from Section~\ref{sec:model}, the received signal at user $\us$ from BS $\x$ is given by
\begin{equation*}
{\bf y}_\us =  {\bf H}_{\x,\us}\rfprecoder\bbprecoder {\bf s}_u +  {\bf n} + \text{OCI},
\end{equation*}
where ${\bf s}_u$ are transmit symbols of dimension $\mathrm{N}_s\times 1$ with energy per symbol equal to $\mathrm{P}/\mathrm{N}_s$, ${\bf n}$ is the noise power (complex Gaussian with zero mean and variance $\sigma_n^2$). 
We assume equal power allocation to all streams. After RF and baseband combining at the receiver, the processed signal is of the form ${\ubbprecoder}^*{\urfprecoder}^*{\bf y}_u$.
When Gaussian symbols are transmitted over the mmWave channel, the spectral efficiency can be at most \cite{Aya14} 
\begin{equation*}
\label{eq:speceff}
r = \log_2\bigg|{\bf I_{\mathrm{N}_s}} + \frac{\power{}}{\mathrm{N}_s} {\bf R}^{-1}_{\bf n} {\bf H}_\text{eff} {\bf H}^{*}_\text{eff}\bigg|,
\end{equation*}
where ${\bf H}_\text{eff} = {\ubbprecoder}^{*} {\urfprecoder}^{*}  {\bf H}_{\x,\us}\rfprecoder\bbprecoder$, ${\bf I}_{\mathrm{N}_s}$ is an identity matrix of rank $\mathrm{N}_s$ and 
\begin{equation*}
{\bf R}_{\bf n} = \sigma^2_n {\ubbprecoder}^{*} {\urfprecoder}^{*} \urfprecoder\ubbprecoder.
\end{equation*}
%This spectral efficiency can be achieved when ${\bf R}^{-1}_{\bf n} {\bf H}_\text{eff} {\bf H}^{*}_\text{eff}$ is a diagonal matrix, else there would be cross talk interference amongst the streams. 
We use the near optimal precoding-combining algorithm proposed in \cite{Aya14} for our simulations. Assuming that an equal fraction of resource is allocated to each UE connected to a BS, the per user rate is defined as $
R_{\x,\us} = \omega_\mathrm{SM}{\mathrm{B}} r/{N_\x}$, where $\omega_\mathrm{SM}$ is the efficiency factor for SM and recalling that $N_\x$ is the total number of users associated with the BS at $\x$. Note that similar to the MU case, $ \omega_\mathrm{SM}$ is dependent on several network parameters like number of antennas, the channel parameters, number of streams, etc. but we drop this in the notation for convenience. Sum rate is defined as the total bits per second transmitted by a BS in Section~\ref{sec:MUMIMO}. Based on this definition, we define the sum rate for the SM enabled mmWave network to be $R_\x = \omega_\mathrm{SM} \mathrm{B} r$. 

\subsection{Heuristic Comparison of Coverage and Rate for MU-MIMO and SM}
\label{sec:heuristic}
In this section, we denote the $\SNR$ with a superscript $\mathrm{SM}$ and $\mathrm{MU}$ to identify spatial multiplexing and MU-MIMO. Round robin scheduling and  $\omega_\mathrm{MU} = \omega_\mathrm{SM} = 1$ will be assumed in this section.
Recall from Lemma~\ref{lem:precodcombin} that for a large number of antennas the eigenvalues of the channel matrix $\mathrm{\bf H}_{\x,\us}$ converge to $\frac{\NBS \NMS  |\gamma_{i,\x,\us}|^2}{L(\x,\us)^{-1}}$. Thus,  the ratio $\frac{\SNR_{i,\x,\us}^{\mathrm{SM}}}{\mathrm{G}}\stackrel{d}{\to} \frac{|\gamma_{i,\x,\us}|^2 L(\x,\us)^{-1} }{\eta_{\x,\us}\mathrm{N}_s}$. From Remark~\ref{rem:convergence}, the ratio $\frac{\SNR_{\x,\us}^{\mathrm{MU}}}{\mathrm{G}}\stackrel{d}{\to} \frac{|\gamma_{i_m,\x,\us}|^2 L(\x,\us)^{-1} }{\eta_{\x,\us}\mathrm{N}_s}$, where $i_m = \arg\max_i \gamma_{i,\x,\us}$.  Since 
$\gamma_{i,\x,\us}\stackrel{\text{st}}{\leq}\max_{i}\gamma_{i,\x,\us}$, we can conclude that in the limit as $\NBS\to\infty$ and $\NMS\to\infty$, $\frac{\SNR_{i,\x,\us}^{\mathrm{SM}}}{\mathrm{G}}\stackrel{\text{st}}{\leq}\frac{\SNR_{\x,\us}^{\mathrm{MU}}}{\mathrm{G}}$ for all $i\in\{1,\ldots,\mathrm{N}_s\}$.

The above discussion hints that for many antennas at BS and UE, the $\SNR$ with MU-MIMO stochastically dominates the $\SNR$ on each stream of SM. If the network were to be noise-limited, the per user and sum rates with MU-MIMO will be higher than SM for a large number of antennas and the same number of streams. Now, let us consider how this result might be affected by OCI. As the number of antennas become large, the effect of zero forcing on the interfering streams is negligible for both MU-MIMO and SM (since the virtual channel approximation in \cite{Say02} starts to more closely model the actual channel). Thus, if the number of streams transmitted by the BS with SM and MU-MIMO are the same, the interference statistics with MU-MIMO and SM would be similar and one would expect that MU-MIMO still outperforms SM for a large number of antennas at BSs and UEs. 

For a finite number of antennas the ZF penalty may be non-negligible. It is expected that the ZF penalty with SM will be less than MU-MIMO since there are more sidelobes that need to be suppressed with MU-MIMO. Thus, the above $\SNR$ dominance result holds given that the number of antenna is large enough such that the effect of the smaller ZF penalty with SM does not reverse the inequalties. For a finite number of antennas, it is neither obvious nor analytically tractable to conjecture as to whether the per user and sum rate of SM would dominate or whether MU-MIMO would. We, thus, rely on Monte Carlo simulations for SM while comparing with our validated analytical model for MU-MIMO and SU-BF. %On a different note, it is important to consider that SU-BF suffices per UE for employing MU-MIMO whereas SM would require the UEs to perform some form of inter-stream interference cancellation using hybrid precoding or any other technique.  %2.5 pages
\section{Numerical Results}
\label{sec:results}
In this section, we first validate the $\SNR$, $\SINR$ and rate coverage analysis from Section III. Next, we compare the per user and sum rate for SU-BF, MU-MIMO and SM with fixed number of BSs per unit area as well as fixed power consumption per unit area.  The default parameters used for generating the results are given in Table~\ref{tab:parameters}. The efficiency factors $\omega_\mathrm{MU}$ and $\omega_\mathrm{SM}$ are implementation specific and estimating these is not the focus of this study. Thus, we set the efficiency parameters to 1 and use Definition~\ref{def:releffi}  for quantifying the allowable relative efficiency.
\begin{table}
\caption{Simulation parameters}
\centering
\label{tab:parameters}
\begin{tabulary}{\columnwidth}{|L | C| L|C|}
\hline
{\bf Parameter} & {\bf Value(s)} & {\bf Parameter} & {\bf Value(s)}\\\hline
$\mathrm{f}_c$ & 73 GHz \cite{Ghosh14} &
$\res$  & 1 GHz \cite{Ghosh14}\\\hline
$\plos,\;\dmax$ & $0.11,\; 200$ m \cite{SinJSAC14}& $\sigma^2_{n}$ & $-174 + 10\log_{10}\res + 10$ dBm \\\hline
$\alpha$ (LOS, NLOS) & 2, 3.3 \cite{Ghosh14}&
$\xi$ (LOS, NLOS) & 5.2, 7.6 \cite{Ghosh14}\\\hline
$\userdnsty$ & 500/km$^2$ &
$\dnsty$ & 60/km$^2$\\\hline
$\NMS$ & 16 \cite{Akd14, BaiAlk14}&
$\NBS$ & 64\cite{Akd14,BaiAlk14}\\\hline
  $\power{}$ & 30 dBm \cite{roh14}
& $\eta_\mathrm{L}$, $\eta_\mathrm{N}$ & 1,3 \cite{Akd14,Sam16,Alk14}\\\hline
\end{tabulary}
\end{table}
%\footnotetext{Experimental observations at 28 GHz and 73 GHz in \cite{Akd14,Sam16} imply that there are on an average 1 to 2 dominant clusters for LOS and 2 to 3 clusters for NLOS links.}
\ifSC
\begin{figure*}
  \centering
\subfloat[Comparison of $\SNR$ analysis with $\SINR$ simulation.]
{\label{fig:validation1}{\includegraphics[width=0.5\columnwidth, trim = 0.9cm 0.3cm 1cm 0.6cm,clip]{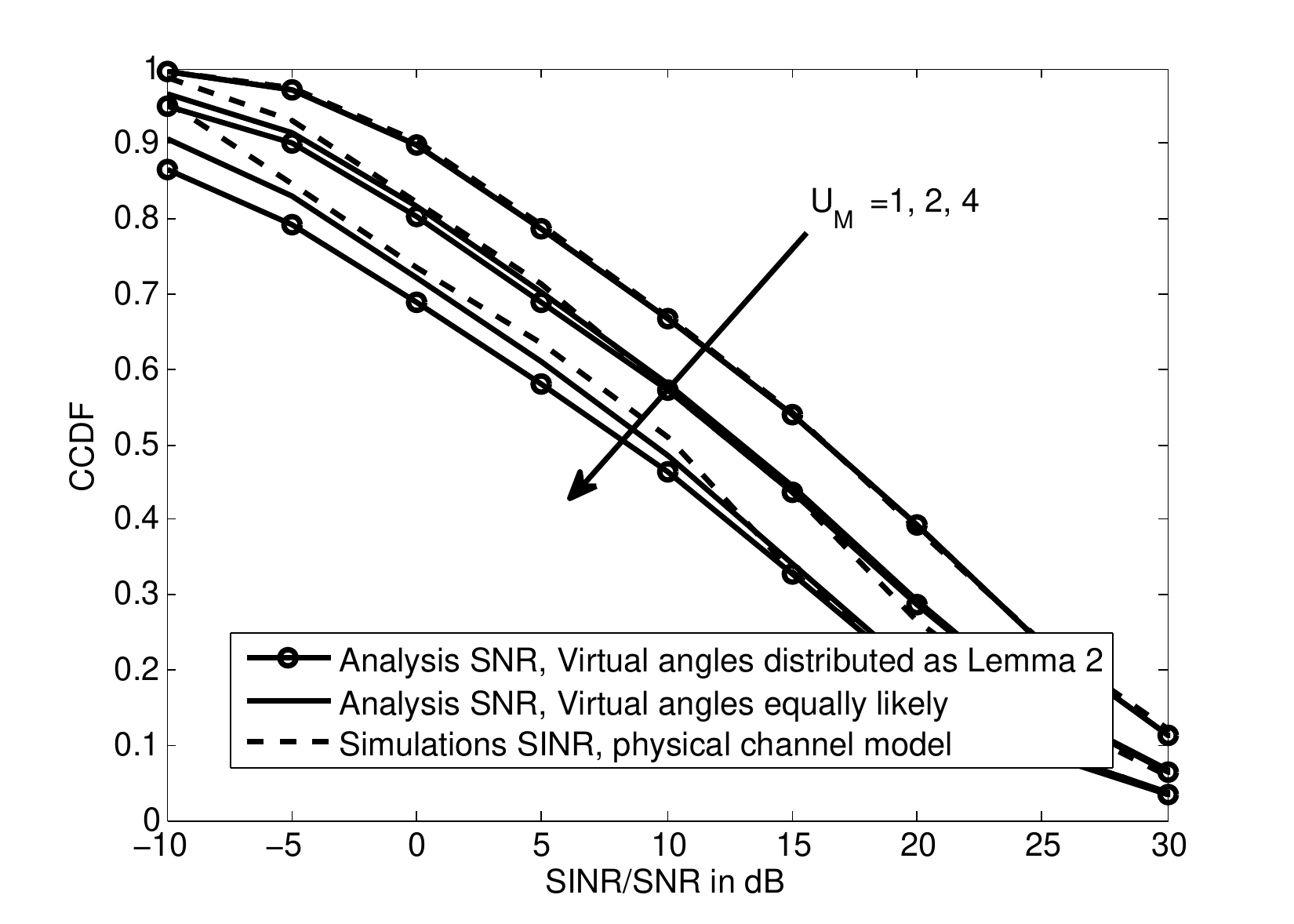}}}
\subfloat[Comparison of per user rate]
{\label{fig:validation2}{\includegraphics[width= 0.5\columnwidth, trim = 0.9cm 0.3cm 1cm 0.6cm,clip]{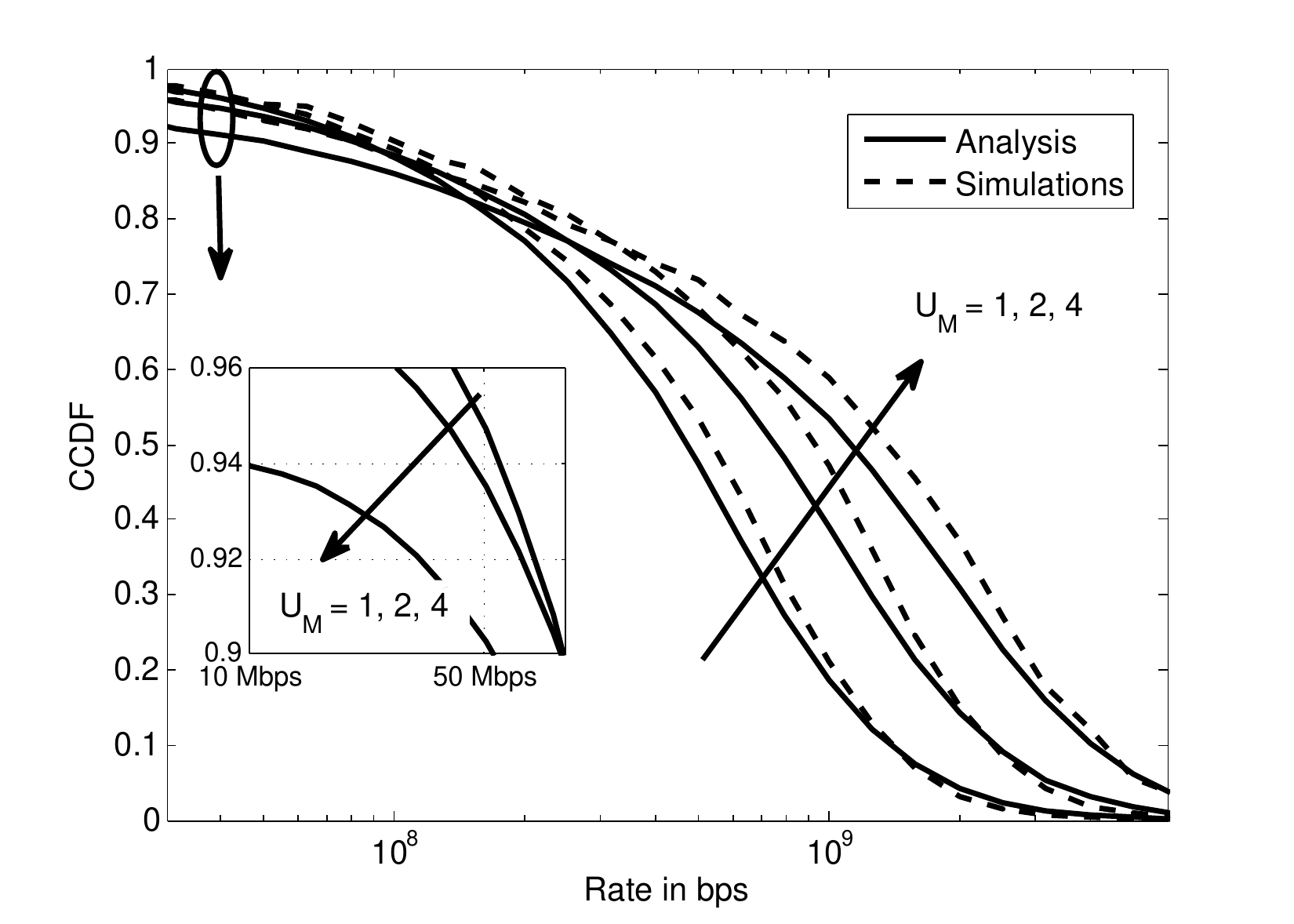}}}
\caption{Validation of $\SNR$ analysis in noise-limited scenario shows a tight match with the physical channel model simulations. Tradeoff between $\SINR$ and rate coverage is also shown with MU-MIMO.}
 \label{fig:validation}
\end{figure*}
\else
\begin{figure*}
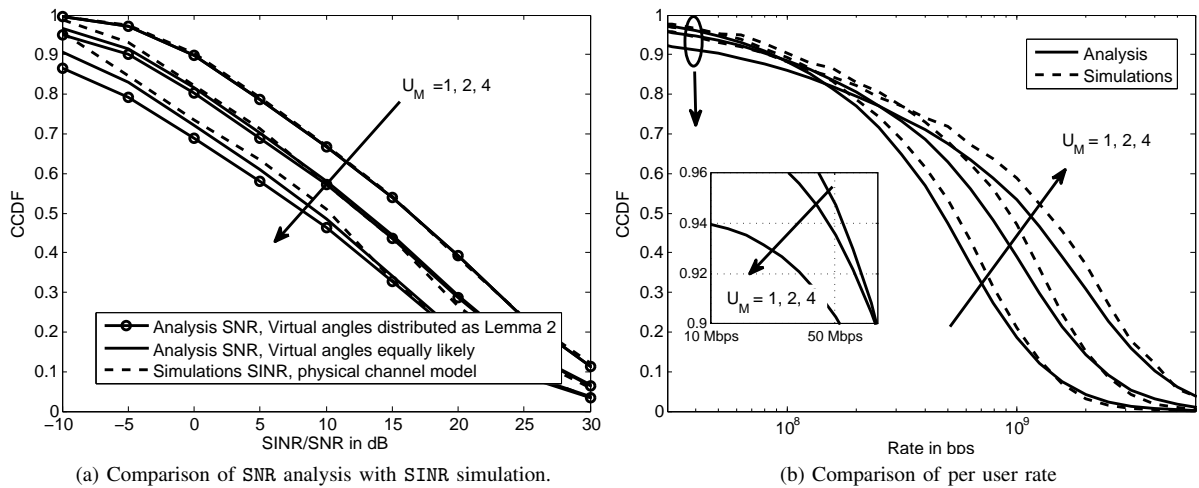

  \centering
\subfloat[Comparison of $\SNR$ analysis with $\SINR$ simulation.]
{\label{fig:validation1}{\includegraphics[width=\columnwidth, trim = 0.9cm 0.3cm 1cm 0.6cm,clip]{Fig1}}}
\subfloat[Comparison of per user rate]
{\label{fig:validation2}{\includegraphics[width= \columnwidth, trim = 0.9cm 0.3cm 1cm 0.6cm,clip]{Fig2}}}
\caption{Validation of $\SNR$ analysis in noise-limited scenario shows a tight match with the physical channel model simulations. Tradeoff between $\SINR$ and rate coverage is also shown with MU-MIMO.}
 \label{fig:validation}
\end{figure*}
\fi

\subsection{Coverage and Rate with MU-MIMO: Validation and Trends}
\subsubsection{Cases Where Interference is Negligible}
Fig.~\ref{fig:validation1} shows the validation of the $\SNR$ coverage formula in Theorem~\ref{thm:sinr}. As can be seen from the figure, the analysis is a tight approximation with the simulations using the physical channel model even when the virtual angles are equally likely, in which case we have much simplified analytical expressions as compared to when the distribution is as given in Lemma~\ref{lem:angle}.  Henceforth, all analysis plots will be with equally likely virtual angles. As expected, the match loosens as $\NRF$ approaches $\NBS$ and $\NMS$. With increasing $\NRF$, the $\SINR$ coverage decreases since the transmit power is split amongst the multiple users served by the BS. However, as seen from Fig.~\ref{fig:validation2}, the median and peak per user rate increases with MU-MIMO. This is due to the fact that in round robin scheduling, each user connected to BS at $x$ now gets $\min\left(\NRF, N_x\right)$ times more slots to transmit. 
A re-interpretation of the above result can be made in terms of minimum allowable efficiencies. For example, $\mathcal{O}_{\{\NRF = 2\},\{\NRF = 1\}}(0.5) = 62.67\%$  and $\mathcal{O}_{\{\NRF = 4\},\{\NRF = 1\}}(0.5) = 42.73\%$. This means that if the efficiency of implementing MU-MIMO with $\NRF = 2$ is at least $62.67\%$ of the efficiency with $\NRF = 1$, then it is beneficial to employ MU-MIMO with $\NRF=2$ over SU-BF in terms of the median rates.

Since $\SINR$ decreases with $\NRF$, the trend for cell edge rates is exactly opposite to peak and median rates. Note that in \cite{Voo14}, it was shown that cell edge rates can improve with MU-MIMO. However, the main difference in their model is the user selection and scheduling. In \cite{Voo14}, there is a high priority user scheduled in a time slot and additional users are served using MU-MIMO only if the expected sum proportional fair metric does not increase due to addition of more users. This protects the rates achieved by cell edge users. The result in Fig.~\ref{fig:validation2}, thus, highlights the importance of user selection and scheduling to protect the rates achieved by cell edge users with multiuser transmission. 

%\begin{figure}
%\centering
%\includegraphics[width=0.5\columnwidth]{Fig3.eps}
%\caption{Impact of multipath combining on $\SNR$ coverage. }
%\label{fig:varyeta1}
%\end{figure}
%Before, we validate the $\SINR$ and rate coverage for the case when interference effects aren't negligible, we make a quick observation on variation of $\SNR$ coverage with number of multipath components. Fig.~\ref{fig:varyeta1} shows that with increasing multipath components the $\SNR$ coverage decreases. This is due of the fact that the channel gain is distributed over multiple AOA/AOD pairs. But the highly directional mmWave beamforming is done over only one such pair as per the system model in Section~\ref{sec:MUMIMO}.  Note that if SM or a diversity technique like maximal ratio combining is employed, beamforming weights needs to be chosen such that there are beamforming gains over multiple AOA/AOD pairs. Fig.~\ref{fig:varyeta1} shows that some form of diversity combining (like selection combining in this work) is highly essential to give significant boost in $\SNR$ coverage and mitigate the leakage in channel gain over unintended directions. Note that selection combining was done implicitly in this work, based on the precoding/combining algorithm in \cite{Alk14}. 

\ifSC
\begin{figure*}
  \centering
\subfloat[Validation of $\SINR$, $\eta_\mathrm{L} = \eta_\mathrm{N}=1$, $\NBS = 64$.]
{\label{fig:Fig5}{\includegraphics[width=0.5\columnwidth, trim = 0.9cm 0.3cm 1.0cm 0.6cm,clip]{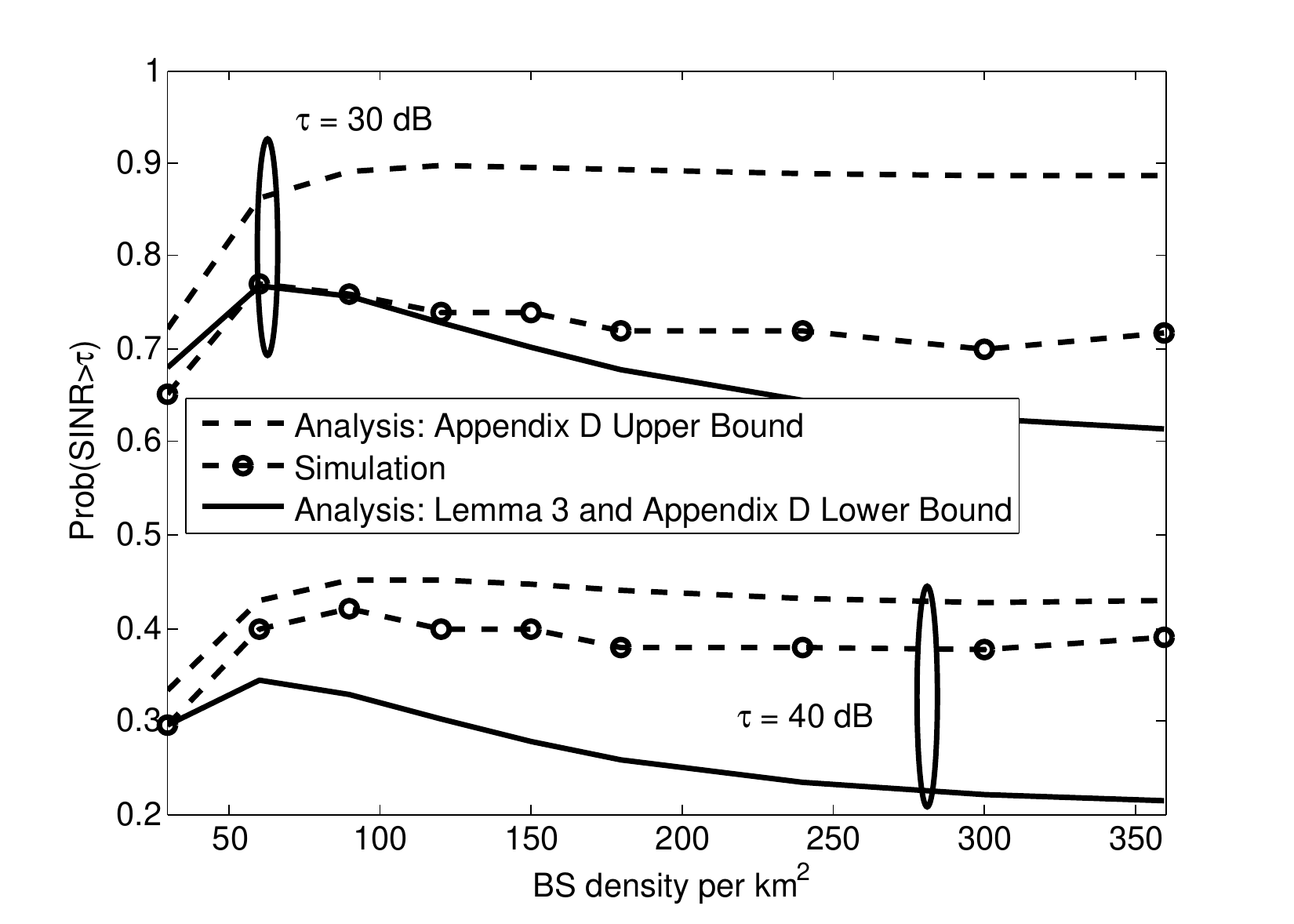}}}
\subfloat[Validation of $\rho_{(.)}$ and bounds on Laplace functional for $\eta_\mathrm{L} = 2$, $\eta_\mathrm{N}=3$.]
{\label{fig:Fig51}{\includegraphics[width= 0.5\columnwidth, trim = 0.9cm 0.3cm 1.0cm 0.6cm,clip]{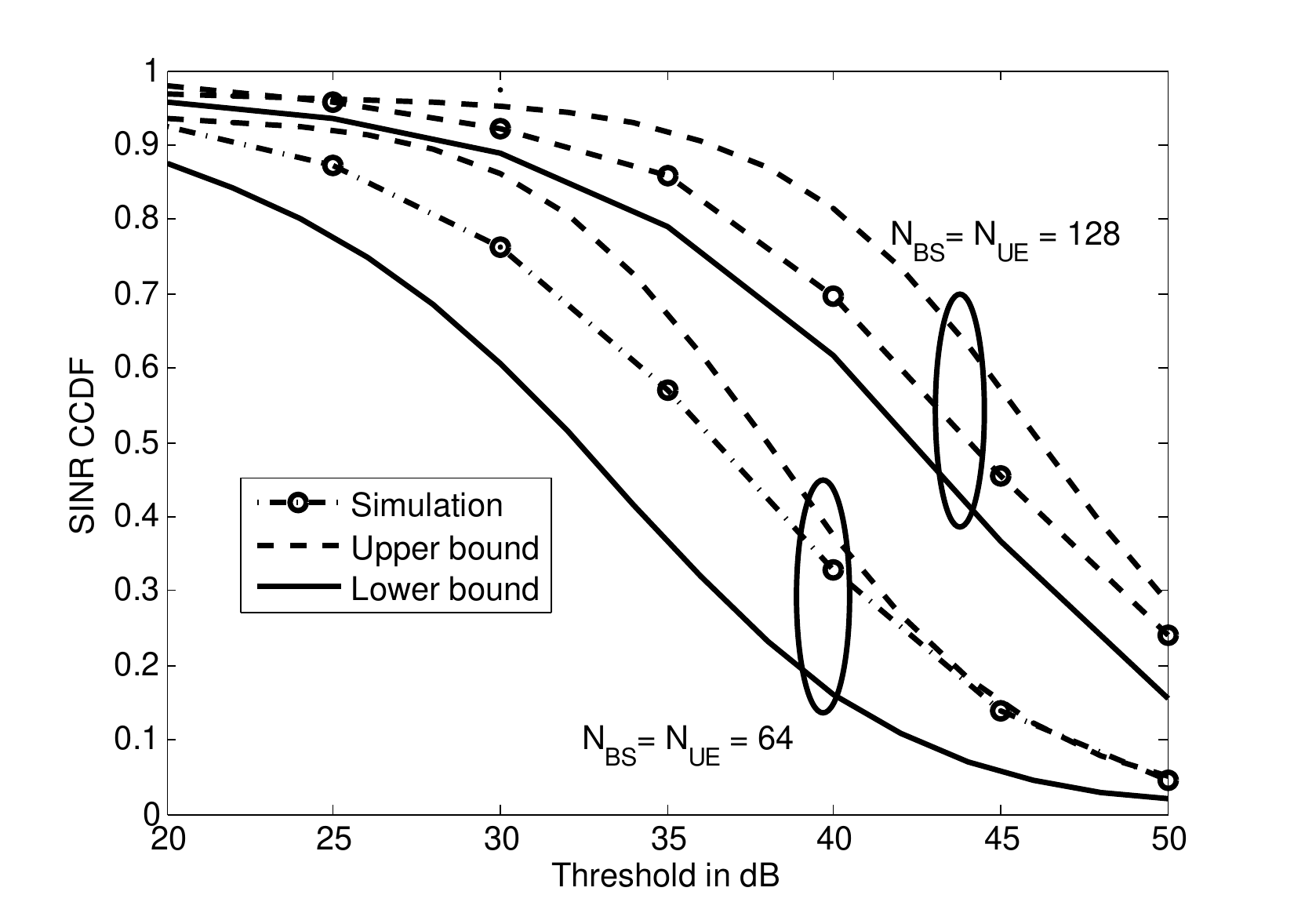}}}
\caption{Validation in interference limited setting for $\NRF = 4$ shows that the upper and lower bounds are within $\pm 5$ dB of the actual simulations. $\SINR$ coverage has a non-monotonic trend with BS density.}
 \label{fig:validationint}
\end{figure*}
\else
\begin{figure*}
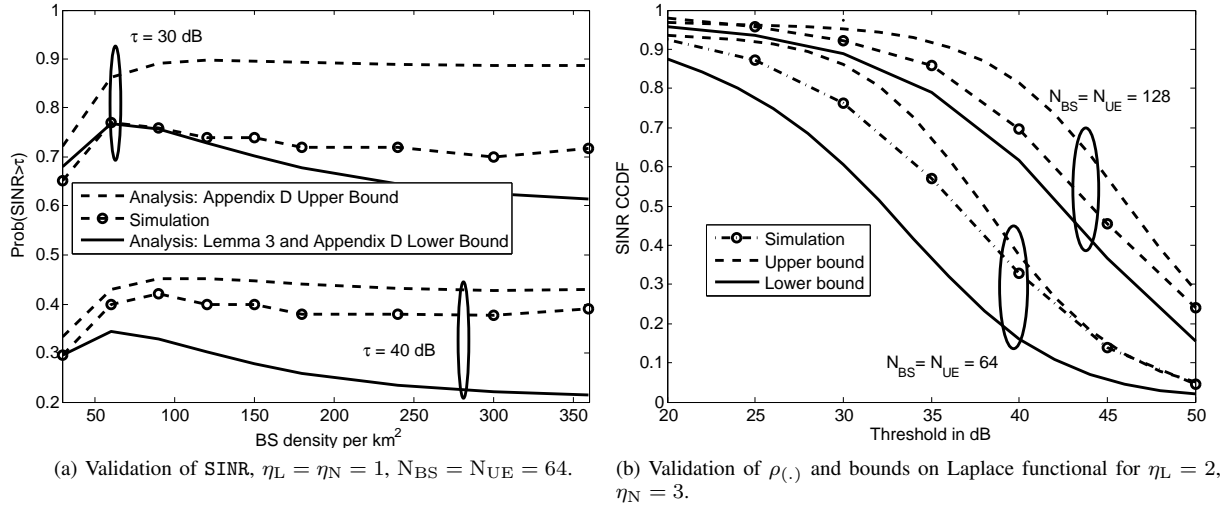

  \centering
\subfloat[Validation of $\SINR$, $\eta_\mathrm{L} = \eta_\mathrm{N}=1$, $\NBS = \NMS = 64$.]
{\label{fig:Fig5}{\includegraphics[width=\columnwidth, trim = 0.9cm 0.3cm 1.0cm 0.6cm,clip]{Fig3}}}
\subfloat[Validation of $\rho_{(.)}$ and bounds on Laplace functional for $\eta_\mathrm{L} = 2$, $\eta_\mathrm{N}=3$.]
{\label{fig:Fig51}{\includegraphics[width= \columnwidth, trim = 0.9cm 0.3cm 1.0cm 0.6cm,clip]{Fig4}}}
\caption{Validation in interference limited setting for $\NRF = 4$ shows that the upper and lower bounds are within $\pm 5$ dB of the actual simulations. $\SINR$ coverage has a non-monotonic trend with BS density. }
 \label{fig:validationint}
\end{figure*}
\fi
\ifSC
\begin{figure*}
  \centering
\subfloat[Sparse multipath: MU-MIMO outperforms except for very low loads. $\eta_\mathrm{L} = 2$, $\eta_\mathrm{N} = 3$ ]
{\label{fig:lowload}\includegraphics[width=0.5\columnwidth, trim = 0.9cm 0.3cm 1cm 0.6cm,clip]{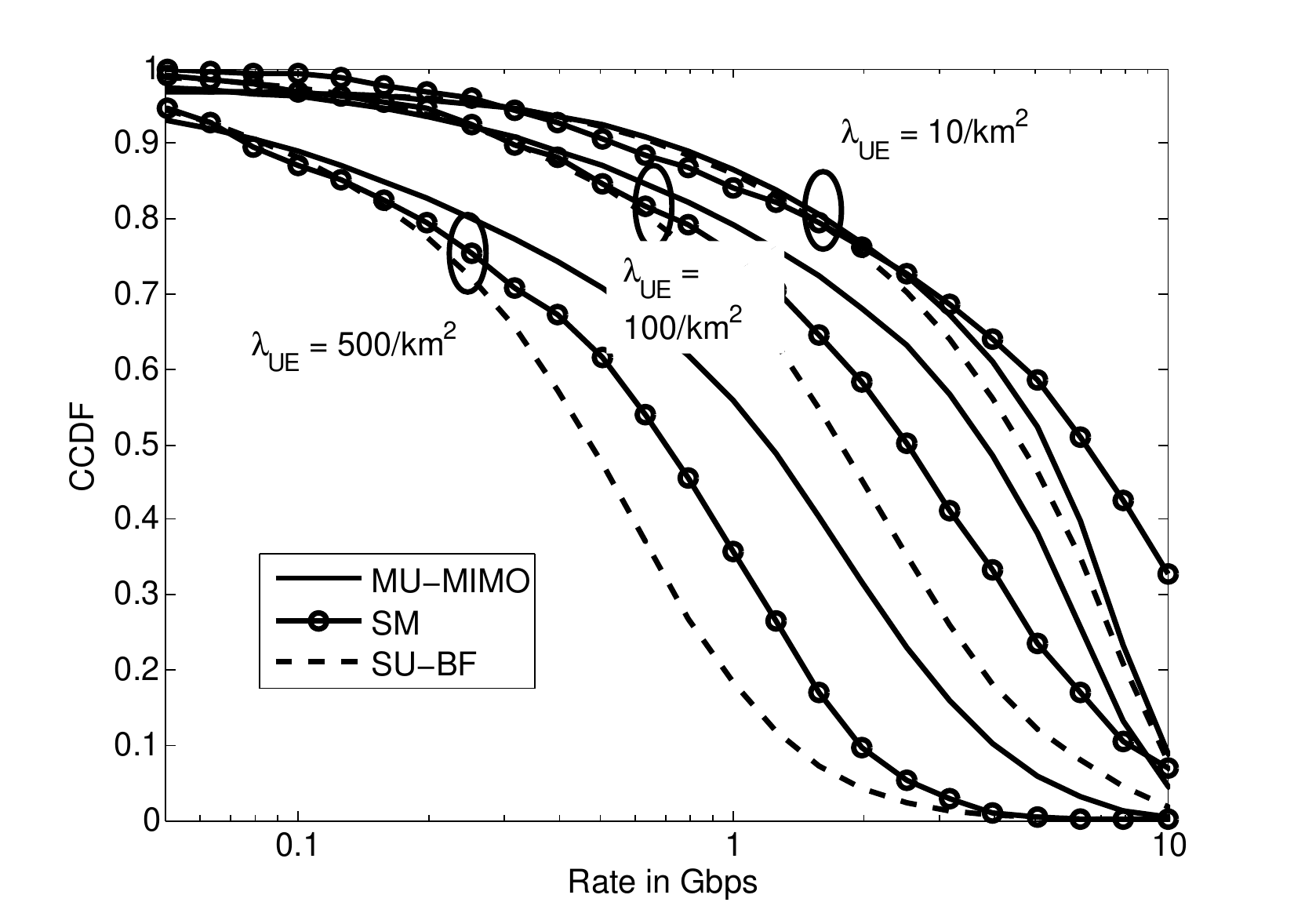}}
\subfloat[Enabling 4 stream SM for $\lambda_\mathrm{UE} = 100/$km$^2$ shows benefit over MU-MIMO for median/peak rates. $\eta_\mathrm{L} = 10$, $\eta_\mathrm{N} = 12$]
{\label{fig:highmultipath}\includegraphics[width= 0.5\columnwidth, trim = 0.9cm 0.3cm 1cm 0.6cm,clip]{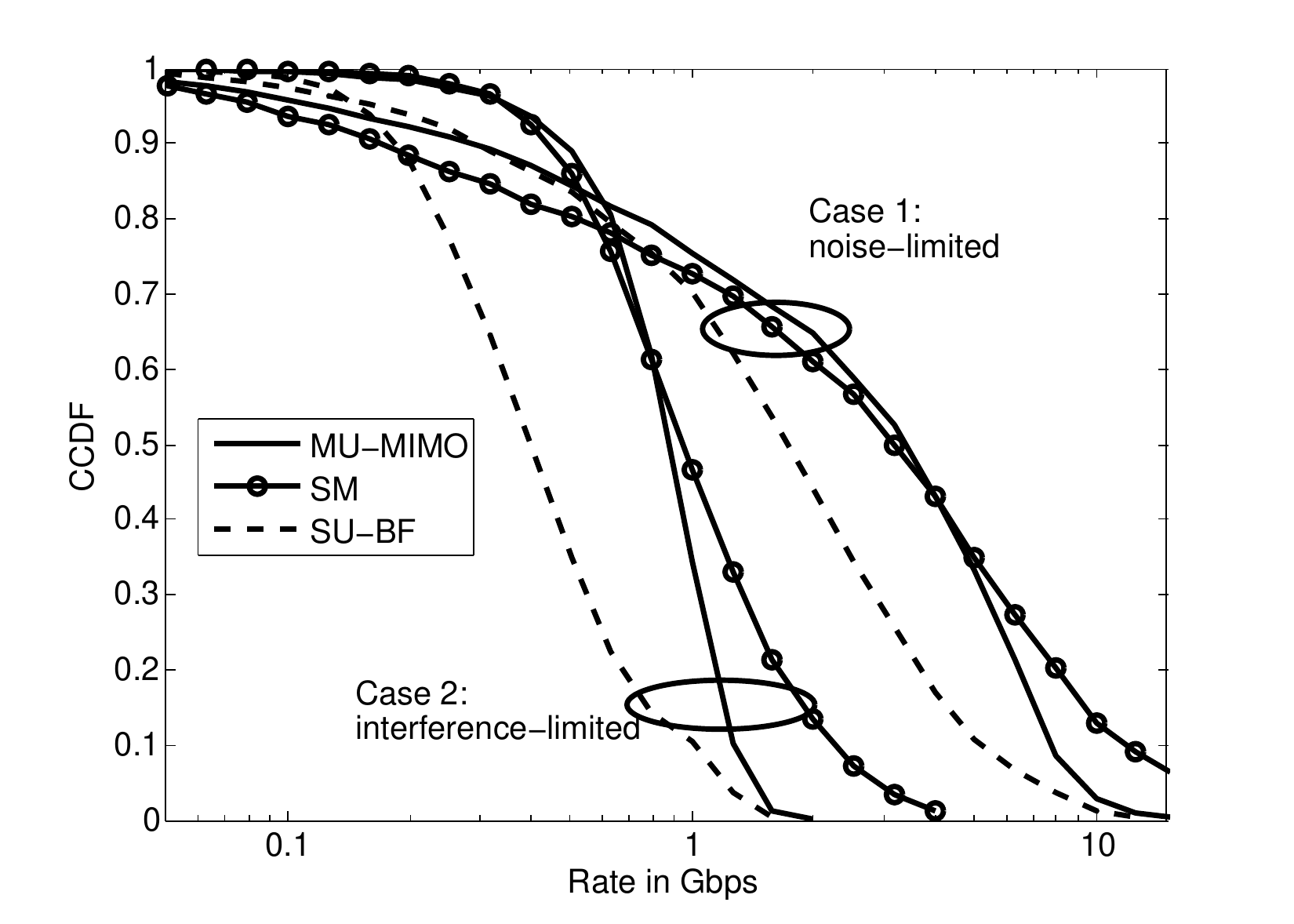}}
\caption{Comparison of MIMO techniques with fixed BS density as well as fixed number of antennas per BS/UE. $\mathrm{U}_\mathrm{M} = 4$}
 \label{fig:compare1}
\end{figure*}
\else
\begin{figure*}
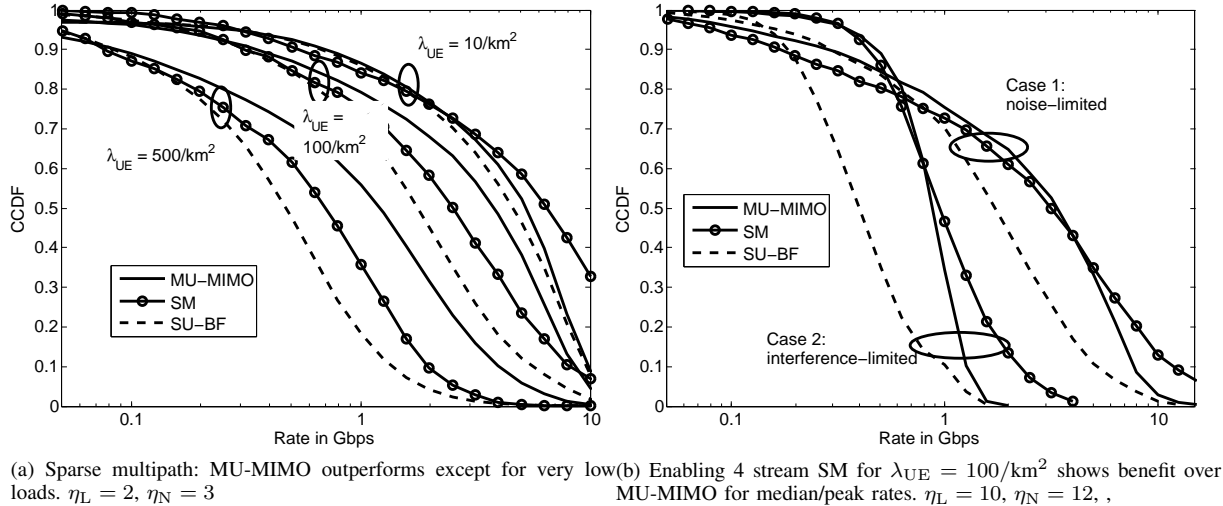

  \centering
\subfloat[Sparse multipath: MU-MIMO outperforms except for very low loads. $\eta_\mathrm{L} = 2$, $\eta_\mathrm{N} = 3$ ]
{\label{fig:lowload}\includegraphics[width=\columnwidth, trim = 0.9cm 0.3cm 1cm 0.6cm,clip]{Lowload}}
\subfloat[Enabling 4 stream SM for $\lambda_\mathrm{UE} = 100/$km$^2$ shows benefit over MU-MIMO for median/peak rates. $\eta_\mathrm{L} = 10$, $\eta_\mathrm{N} = 12$, , ]
{\label{fig:highmultipath}\includegraphics[width= \columnwidth, trim = 0.9cm 0.3cm 1cm 0.6cm,clip]{Highmultipath}}
\caption{Comparison of MIMO techniques with fixed BS density. $\mathrm{U}_\mathrm{M} = 4$}
 \label{fig:compare1}
\end{figure*}
\fi
\subsubsection{Cases Where Interference is Not Negligible}
Fig.~\ref{fig:Fig5} shows the validation of $\SINR$ coverage formula in Theorem~\ref{thm:sinr2} for single path scenario. In order to present a case where interference effects are not negligible we consider a network at 28 GHz band with 200 MHz bandwidth, a less blocked scenario with $\plos = 0.5$ for $\mathrm{D}=  200$m and much higher $\userdnsty = 1000/$km$^2$. As per discussion in Section~\ref{sec:MUMIMO}, $\rho_\mathrm{BS}=1/(\mathrm{N_\mathrm{BS}} \sin(0.244))$ and $\rho_\mathrm{UE} = 1/(\mathrm{N_\mathrm{UE}} \sin(0.244))$.  Fig.~\ref{fig:Fig5} shows that increasing BS density does not necessarily improve coverage. This trend is similar to that observed in \cite{BaiHea14} and shows the presence of an optimal BS density in terms of $\SINR$ coverage. Approximate analytical results in Lemma~\ref{lem:int} and Appendix~\ref{app:int2} capture the essential non-monotonic trend shown with the simulations. Fig.~\ref{fig:Fig51} further validates the analysis in Appendix~\ref{app:int2} for multipath scenario as well as shows a decreasing gap with the physical model simulations as the number of antennas grows large. Both these plots build confidence in the analysis and derived insights. Using analysis, it can be found that optimum BS density for $\NRF = 1, 2$ and $4$ decreases as $82, 72$ and $63$ BSs/km$^2$. Thus, with increasing $\NRF$ the optimum BS density reduces due to increasing interference in the network. %Although not shown in Fig.~\ref{fig:Fig5}, as BS density approaches UE density or becomes greater than that, the coverage again starts increasing since interference in the network is saturated. 
%Another thing to note is that ultra-dense deployment of mmWave networks with large number of antenna arrays and moderate bandwidth gives significantly higher $\SINR$ coverage, unlike large bandwidth networks considered in Fig~\ref{fig:validation1}. Fig.~\ref{fig:Fig5}  shows that even an $\SINR$ on the order of $30$ to $40$ dB may be common. In order to convert these high spectral efficiencies into user perceived throughput, a major challenge would be to support higher order constellations. However, since large bandwidth available in mmWave bands serves as a brute force solution to offer high data rates, it is likely that extreme densification or very large antenna arrays at UEs would not be required at least for initial deployments.  One caveat here is that we have not considered human body blockages or foliages in our analysis which might give higher estimates of $\SINR$ in this work and it is desirable to incorporate these effects in future to see if these insights hold.

%
%\begin{table}
%\caption{Optimum BS density versus $\NRF$, with $\eta_\mathrm{L} = \eta_\mathrm{N} = 1$}
%\label{tab:optimumdensity}
%\centering
%\begin{tabulary}{\columnwidth}{|C | C|}
%\hline
%{\bf $\NRF$} & {\bf Optimum BS density (per km$^2$)}\\\hline
%1 & 82 \\\hline
%2 & 72 \\\hline
%4 & 63\\\hline
%\end{tabulary}
%\end{table}
\subsection{Comparing Per User and Sum Rate for SU-BF, MU-MIMO and SM}
The gains with SM and MU-MIMO are fundamentally driven by distinct network parameters. For example, having more number of multipaths (or larger $\eta_\mathrm{L}$ and $\eta_\mathrm{N}$) increases the rank of the channel and thus enables transmitting more number of streams with single user SM, given that there are enough RF chains at the transmitter and the receiver. However, this does not necessarily help in having more multi-user streams. On the other hand, having low load reduces the possible gain with MU-MIMO even if each BS is equipped with a large number of RF chains due to the fact that there are not many users to schedule simultaneously per BS. This does not however affect SM in terms of the number of streams per user. Thus, sufficiently low load and high multipaths may cause SM to outperform MU-MIMO given that there are enough RF chains at the BSs and UEs. This can be seen in Figures~\ref{fig:lowload} and \ref{fig:highmultipath}. The plots for MU-MIMO and SU-BF in Figure~\ref{fig:lowload} are with analysis. The plots for SM in Figure~\ref{fig:lowload} and the entire Figure~\ref{fig:highmultipath} is using Monte-Carlo simulations. Note that our analytical model is valid for $\eta_\mathrm{L},\eta_\mathrm{N}\gg \NMS$ and not for  $\eta_\mathrm{L},\eta_\mathrm{N}$ close to $\NMS$, which is the case in Figure~\ref{fig:highmultipath}. 
\ifSC
\begin{figure}
\centering
\includegraphics[width= 0.5\columnwidth, trim = 0.9cm 0.3cm 1cm 0.6cm,clip]{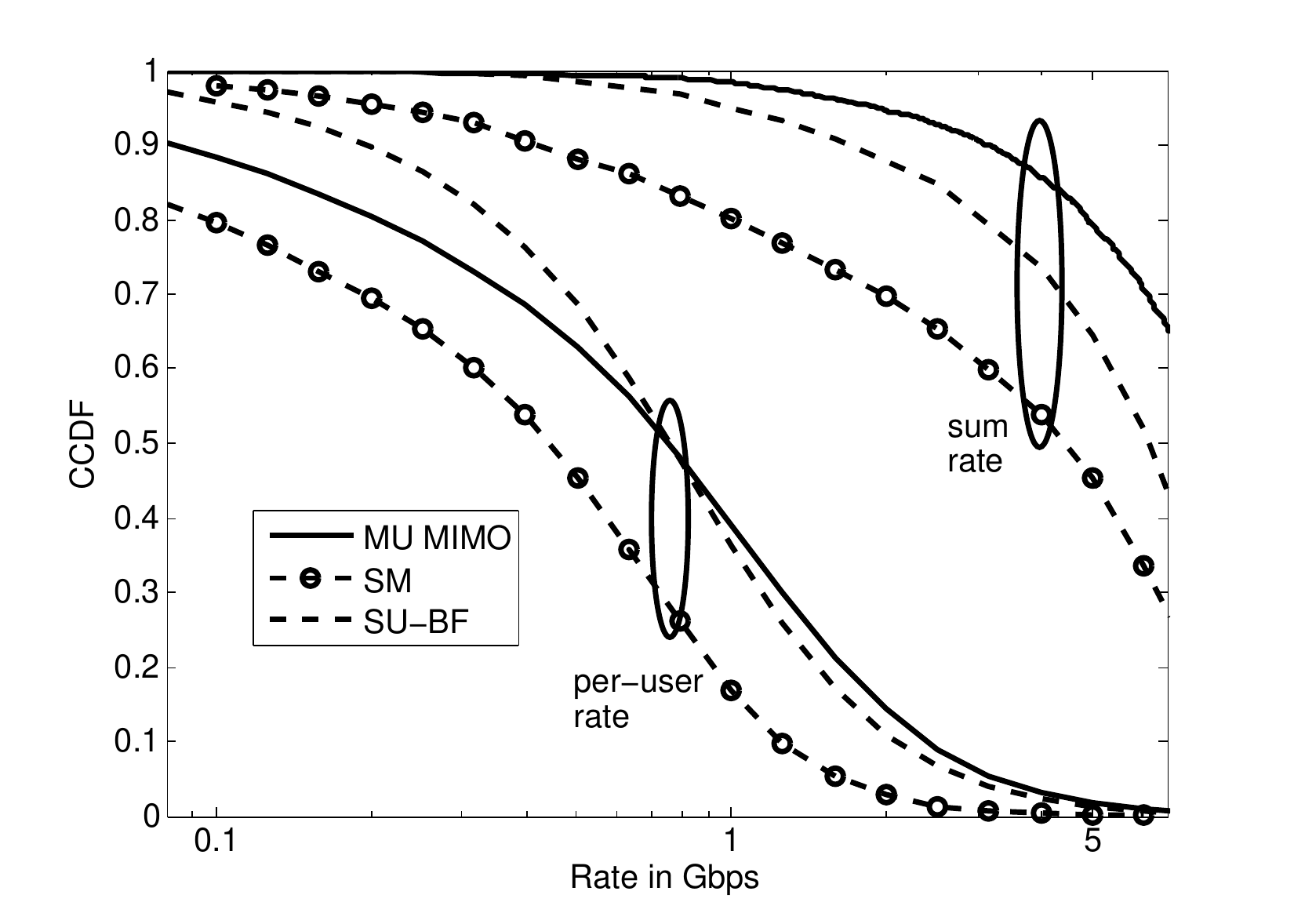}
\caption{Comparison of per user and sum rate with fixed power consumption. A denser SU-BF network outperforms MU-MIMO and SM in terms of cell edge per user rates but MU-MIMO performs the best in terms of sum rates. }
 \label{fig:compare2}
\end{figure}
\else
\begin{figure}
\centering
\includegraphics[width= \columnwidth, trim = 0.9cm 0.3cm 1cm 0.6cm,clip]{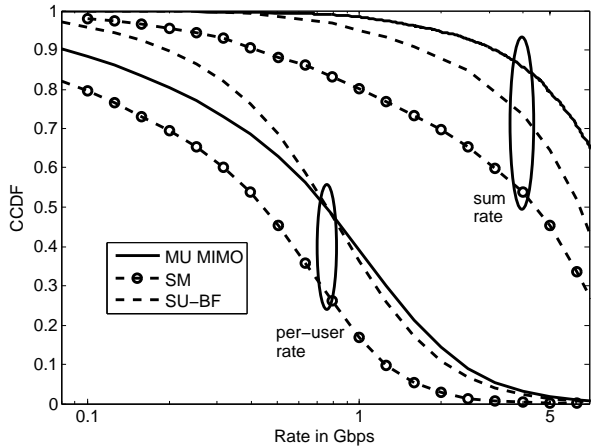}
\caption{Comparison of per user and sum rate with fixed power consumption. A denser SU-BF network outperforms MU-MIMO and SM in terms of cell edge per user rates but MU-MIMO performs the best in terms of sum rates. }
 \label{fig:compare2}
\end{figure}
\fi
Figure~\ref{fig:lowload} shows that for moderate and low user densities (which corresponds to $\lambda_\mathrm{UE} = 500/$km$^2$ and $\lambda_\mathrm{UE} = 100/$km$^2$) MU-MIMO outperforms SM and SU-BF. However, for very low load (corresponds to $10$ UEs/km$^2$) SM outperforms MU-MIMO. This result is due to the fact that although SM can offer 2 streams per user but MU-MIMO cannot provide gains since per km$^2$ there are only 10 users that can associate with 60 possible BSs and the probability that a BS connects to more than 1 user is very low. Since our analytical model slightly loose estimates for low $\SNR$ users, we compare the cell edge rates using simulations only. The cell edge rates are quite close for the three schemes, although SM and SU-BF slightly outperform MU-MIMO. For low loads, SM is slightly better than SU-BF in terms of cell edge rates. Considering that overhead with MU-MIMO could be the highest, this trend will be more exaggerated afte considering these factors. A better scheduling will be indeed important for protecting cell edge rates with MU-MIMO.

Figure~\ref{fig:highmultipath} shows the impact of high multipath on the comparison insights. As was observed in Figure~\ref{fig:lowload}, MU-MIMO outperformed SM for $\lambda_\mathrm{UE} = 100/$km$^2$ when multipath was low.  For the same network parameters, that lead to a noise-limited case, increasing the multipath to  $\eta_\mathrm{L} = 10$ and $\eta_\mathrm{N} = 12$ gives higher rates with SM for even $30$ percentile users. This is again due to the fact that since there are 4 RF chains at UEs and BSs, SM can support 4 streams per user. However, since there are about 1.7 UEs per BS, BSs can only transmit to about 2 UEs per time slot on an average with MU-MIMO. Further the increased multipath leads to higher ZF penalty for MU-MIMO. Similar trend is observed in the interference-limited scenario ($f_c =28\text{ GHz},\mathrm{B} = 100\text{ MHz}, p_\mathrm{LOS} = 0.5$). Since a low blockage scenario is considered, the 4 streams per UE are LOS links with very high probability. Thus, the gains with SM look slightly exaggerated in the interference-limited case. Also note that having a large multipath as considered here could be a unlikely scenario in outdoor mmWave networks \cite{Akd14} but it is interesting to consider from an analytical perspective. 

A re-interpretation of the above plots can be made in terms of minimum allowable efficiency of MU-MIMO to outperform SM or SU-BF. For example, when $\lambda_\mathrm{UE} = 500/$km$^2$ in Figure~\ref{fig:lowload}, MU-MIMO outperforms SM in terms of median users if its efficiency factor is more than 58$\%$. Similarly, such numbers can be extracted for other plots 
using Definition~\ref{def:releffi}. As mentioned earlier, a separate study on estimating these efficiency factors is needed to make a strong claim on comparison of these MIMO techniques. 
%\begin{figure*}
%  \centering
%\subfloat[Comparison of per user rate.]
%{\label{fig:comparepower1}\includegraphics[width=0.5\columnwidth, trim = 0.9cm 0.3cm 1cm 0.6cm,clip]{Fig7}}
%\subfloat[Comparison of sum rate]
%{\label{fig:comparepower2}\includegraphics[width= 0.5\columnwidth, trim = 0.9cm 0.2cm 1cm 0.6cm,clip]{Fig8}}
%\caption{Per user and sum rate comparison with fixed power consumption per unit area. }
% \label{fig:comparepower}
%\end{figure*}

The above comparison results were for fixed BS density and the same number of antennas across different schemes. However, with an increasing number of RF chains, the power consumed per BS also increases. In the hybrid precoding as shown in  Fig.~\ref{fig:hybridarchitecture}, each RF chain is connected to all antennas through phase shifters. Thus, with increasing number of RF chains the number of phase shifters grows proportionally with the number of antennas, and effectively the power consumption is also increased. Let $\nu(\mathrm{N}_\mathrm{RF})$ denote the ratio of power consumed at a BS with $\mathrm{N}_\mathrm{RF}$ RF chains to a BS with 1 RF chain. 
%Based on \cite{rial15}, the power consumption for a fully-connected architecture with $\mathrm{N}_\mathrm{RF}$ chains can be modelled as 
%\begin{equation*}
%\power{\text{BS}}(\mathrm{N}_\mathrm{RF}) = \power{}+\NBS\mathrm{N}_\mathrm{RF}\power{\text{PS}} + \NBS
%\power{\text{LNA}} 
%+\mathrm{N}_\mathrm{RF} (\power{\text{RF chain}}+\power{\text{DAC}}) + \power{\text{BB}},
%\end{equation*}
%where $\power{}$ is the transmit power of BS, $\power{\text{PS}}$ is the power consumed by the phase shifters, $\power{\text{LNA}}$ is the power consumed by the low noise amplifier, $\power{\text{DAC}}$ is the power consumed by the digital to analog converter,
%$\power{\text{BB}}$ is the power consumed by baseband processing and $\power{\text{RF}}$ is power consumed by RF chains. The values will depend on system parameters like bandwidth and carrier frequency in addition to the details of a particular implementation. Note that $\nu(\mathrm{N}_\mathrm{RF}) = \power{\text{BS}}(\mathrm{N}_\mathrm{RF})/\power{\text{BS}}(1)$. 
A ballpark value of $\nu$ can be found to be $1.38$ for $\NBS=64$ and $\mathrm{N}_\mathrm{RF}=\NRF=2$ based on the power consumption model in \cite{rial15} (refer \cite{KulAlk15} for a discussion on this). 
%\begin{table}
%\caption{Power consumption parameters based on \cite{rial15}}
%\label{tab:power}
%\centering
%\begin{tabulary}{\columnwidth}{|L | C|L | C|}
%\hline
%{\bf Parameter} & {\bf Value (in mW)}& {\bf Parameter} & {\bf Value (in mW)}\\\hline
%$\power{}$ & 1000 &
%$\power{\text{PS}}$ & 20\\\hline
%$\power{\text{LNA}}$ & 20 &
%$\power{\text{RF chain}}$ & 40\\\hline
%$\power{\text{BB}}$ & 200 & $\power{\text{DAC}}$ & 200\\\hline
%\end{tabulary}
%\end{table}
We now scale up the BS density of SU-BF by exactly a factor of $\nu$. Note that UEs need to use only single RF chain for SU-BF and MU-MIMO with hybrid precoding. However, UEs need multiple RF chains for SM with hybrid precoding architecture. Thus, for fair comparison considering power consumption model in \cite{rial15} we reduce the $\NMS$ to 7 for SM. As can be seen from Fig.~\ref{fig:compare2}, the gain in per user data rates with MU-MIMO and SM diminishes or completely vanishes if the SU-BF network has 1.38 times denser deployment on an average. Fig.~\ref{fig:compare2} shows that MU-MIMO still has significantly higher sum rates than for a denser SU-BF network. However, per user cell edge rates with a denser SU-BF network are higher in this case. To quantify the cell edge gains in per user rates $\mathcal{O}_{\mathrm{MU},\mathrm{SU}}(0.95) = 315\%$, which is huge and strengthens our conclusion that a denser SU-BF network outperforms MU-MIMO in terms of cell edge rates. Also note that $\mathcal{O}_{\mathrm{MU},\mathrm{SU}}(0.5) = 99\%$, which implies that most likely even the median gains with SU-BF will be better after incorporating the channel acquisition overheads. However, in terms of sum rates $\mathcal{O}_{\mathrm{MU},\mathrm{SU}}(0.5) = 73\%$ which implies that median rates with MU-MIMO can still be higher as long as the efficiency is more than $73\%$ of SU-BF efficiency. 

%To summarize, if leasing sites for deploying more BSs and providing backhaul is a bottleneck, MU-MIMO and SM can be used for boosting data rates as compared to SU-BF mmWave networks. However, efficient implementations of channel acquisition using its sparsity needs to be done to exploit these gains. Neglecting the overheads due to downlink training and computational complexity, MU-MIMO offers significant gains in sum rates as compared to SM and SU-BF for a fixed density of BSs as well as for fixed power consumption per unit area. However, densifying the SU-BF network keeping power consumption per unit area equal to the MU-MIMO and SM counterparts, helps in getting better $\SINR$ coverage and thus better cell edge rates as well.
% %
% Plots: Fig 5 shows interference
\section{Conclusions and Future Work}
\label{sec:conclusion}
The analytical model in this work demonstrates the utility of the virtual channel approximation to incorporate different precoder and combiner constraints in network level analysis of dense MIMO cellular networks with many antennas. It would be beneficial to get tighter bounds on the Laplace functional of the out-of-cell interference. The analytical model can also be extended to incorporate more realistic cross-polarized uniform planar arrays instead of ULA. Another important issue that needs to be addressed is to incorporate the effects of imperfect channel state information in the analytical model. Since MU-MIMO requires more channel state information at the transmitter, imperfect channel knowledge may affect the performance of MU-MIMO more than SM or SU-BF. It is essential to know whether this would overshadow the benefits of MU-MIMO over SM and SU-BF observed in this paper. %Unlike in sub-6GHz networks, compressed sensing tools can be used for low-complexity channel estimation using the sparseness of the channel\cite{AlkMo14}. This might turn out to benefit MU-MIMO, for which imperfect channel estimation and acquisition overhead have been a major challenge in sub-6GHz networks\cite{Ges07}. 

 %If initial deployments are ultra-dense with SU-BF and the BSs are then upgraded to MU-MIMO in future, it is possible that switching off some BSs might help achieve better $\SINR$ coverage. This is due to the fact that the optimum BS density in terms of $\SINR$ coverage is lesser with MU-MIMO as compared to SU-BF, as shown in Section~\ref{sec:results}. 

%Dynamic time division duplexing and self-backhauling have been suggested for employing low latency dense mmWave cellular networks\cite{nokia15, Ghosh14}. With favorable BS to BS propagation and high beamforming gain, the downlink-to-uplink interference may be very difficult to mitigate\cite{Rangan15}, which may affect the MIMO comparison insights presented in this paper. Further, increased interference with MU-MIMO and sectorization may collapse the data rates, especially the edge rates. Investigating how the choice of MIMO techniques vary with these additional constraints on mmWave networks is, thus, an important open question yet to be addressed. %and references: together 1 page 
\begin{appendix}
\section{Appendices}
\subsection{Derivation of Zero Forcing Penalty in Proposition~\ref{prop:1}}
\label{app:prop1}
For simplicity in notation, let us denote by $\theta^{i}_{j}$ and $\phi^{i}_{j}$ as the AOD and AOA on the $j^{\text{th}}$ path from/to the BS at $x$ under consideration to/from the $i^{\text{th}}$ user, $i\in\{1,\ldots,\mathrm{U}\}$, served by the BS, respectively. $\overline{{\bf H}}_{\x,\us}$ is equal to (\ref{eq:equivalentchannel}) when all of the following events are true.
\begin{itemize}
\item $E_1$ : $\mathrm{\bf a}^{*}_{\mathrm{UE}}(\phi^{k}_{1})\aMS(\phi^{k}_{j})\mathrm{\bf a}^{*}_{\mathrm{BS}}(\theta^{k}_{j})\aBS(\theta^{1}_{1}) = 0$ for all $j\in\{1,\ldots,\eta_k\}$ and $k\in\{2,\ldots,\mathrm{U}\}$.
\item  $E_2$ : $\mathrm{\bf a}^{*}_{\mathrm{UE}}(\phi^{1}_{1})\aMS(\phi^{1}_{j})\mathrm{\bf a}^{*}_{\mathrm{BS}}(\theta^{1}_{j})\aBS(\theta^{k}_{1}) = 0$ for all $j\in\{1,\ldots,\eta_1\}$ and $k\in\{2,\ldots,\mathrm{U}\}$.
\item  $E_3$ : $\mathrm{\bf a}^{*}_{\mathrm{UE}}(\phi^{1}_{1})\aMS(\phi^{1}_{j})\mathrm{\bf a}^{*}_{\mathrm{BS}}(\theta^{1}_{j})\aBS(\theta^{1}_{1}) = 0$ for all $j\in\{2,\ldots,\eta_1\}$.
\end{itemize}
Note that probability of $p_\mathrm{ZF}=1$ is given by $\pr(E_1\cap E_2\cap E_3)$. Using the ON/OFF nature of inner products of beamsteering vectors with virtual channel approximation, we can re-write the above conditions as
\begin{itemize}
\item $\mathcal{E}_1=A_1 \cap A_2$, where $A_1 = \bigcap_{k=2}^{\mathrm{U}}\{ \theta^1_1\neq \theta^{k}_1\}$ and $A_2 = \bigcap_{k=2}^{\mathrm{U}}\bigcap_{j=2}^{\eta_k}\{\phi^k_1\neq \phi^k_j\} \cup \{\theta^k_j\neq \theta^1_1 \}$.
\item $\mathcal{E}_2=A_1\cap A_3$, where $A_3= \bigcap_{k=1}^{\mathrm{U}}\bigcap_{j=2}^{\eta_1}\{\phi^1_1\neq \phi^1_j\} \cup \{\theta^1_j\neq \theta^k_1 \}$.
%\item $E_3=A_3$, where $A_3= \bigcup_{j=2}^{\eta_1} \{\phi^1_1\neq \phi^1_j\} \cup \{\theta^1_j\neq \theta^1_1 \}$.
\end{itemize}
Note that $\pr(E_1\cap E_2\cap E_3)= \pr(\mathcal{E}_1\cap \mathcal{E}_2) = \pr(A_1\cap A_2\cap A_3)$. Conditioning on $\theta^1_1$, $A_2$ is independent of $A_1$ and $A_3$. Using (a) $ \pr(A \cup B) =  \pr(A)+\pr(B)-\pr(A \cap B)$, (b) all distinct AOA or AOD are independently distributed as per the distribution given in Lemma~\ref{lem:angle}, (c)  $\bigcap_{k=1}^{\mathrm{U}}\bigcap_{j=2}^{\eta_1}\{\phi^1_1\neq \phi^1_j\} \cup \bigcap_{k=1}^{\mathrm{U}}\bigcap_{j=2}^{\eta_1} \{\theta^1_j\neq \theta^k_1 \}\subset\bigcap_{k=1}^{\mathrm{U}}\bigcap_{j=2}^{\eta_1}\{\phi^1_1\neq \phi^1_j\} \cup \{\theta^1_j\neq \theta^k_1 \}$ and (d) for a highly dense network, the probability that the BS is serving a LOS UE is $\plos$ since the association region of a BS is almost surely covered by the ball of radius $\mathrm{D}$ centered at the BS, the required lower bound on the probability of $p_\mathrm{ZF}=1$ is derived, also given by $\zeta(.)$.  In order to get the more simplified expression in Remark~1, the term $\mathrm{D}_{j}(.)$ in Proposition~\ref{prop:1} needs to be simplified. For equally likely virtual angles, this can be found using the following Lemma, which we propose. 
\begin{lem}
Pick $U$ numbers that take values in range $\{1,\ldots,N\}$. Repetition of values is allowed and order is important. The probability that the first $U_1$ numbers are mutually exclusive from the remaining $U_2=U-U_1$ is given by $\mathcal{P}$, where
\begin{equation*}
\mathcal{P} = \sum_{d=1}^{U_1} {N \choose d} (N-d)^{U_2} \sum_{i=0}^{d} {d \choose i} (-1)^{i} (d-i)^{U_1}.
\end{equation*}
\end{lem}
The idea is to condition that there are $d$ distinct values in first $U_1$ numbers, in which case the remaining $U_2$ numbers can take values in $(N-d)^{U_2}$ ways. Further the number of ways in which first $U_1$ numbers take $d$ distinct values can be found using inclusion exclusion principle, which is given by the inner summation. 

\subsection{Proof of Theorem~\ref{thm:sinr}}
\label{app:sinr}
Let $l^*_\mathrm{L}$ and $l^*_\mathrm{N}$ denote the points closest to origin in $\mathcal{N}_\mathrm{L}$ and $\mathcal{N}_\mathrm{N}$, respectively.
Using Lemma~\ref{lem:mdnsty}, the probability of associating with a LOS BS is given by 

\begin{align*}
\mathrm{A}_\mathrm{L} \iftoggle{SC}{}{&}= \mathrm{B}_\mathrm{L}\int_{0}^{\infty} \pr\left(l^*_\mathrm{N}>t\right) f_\mathrm{L}(t) \mathrm{d}t\iftoggle{SC}{}{\\&}
=\mathrm{B}_\mathrm{L}\int_{0}^{\infty} \exp\left(-\dnsty \mathrm{M}_\mathrm{N}(t)\right) f_\mathrm{L}(t) \mathrm{d}t.
\end{align*}
Similarly, the probability of associating with NLOS BS is given by 
$
\mathrm{A}_\mathrm{N} =\mathrm{B}_\mathrm{N}\int_{0}^{\infty} \exp\left(-\dnsty \mathrm{M}_\mathrm{L}(t)\right) f_{\mathrm{N}}(t)\mathrm{d}t
$.
Similar to Lemma 3 in \cite{BaiHea14}, the PDF of propagation loss to associated BS given that the association is of type LOS, is given by 
$
\tilde{f}_\mathrm{L}(t) = \frac{\mathrm{B}_\mathrm{L}}{\mathrm{A}_\mathrm{L}}f_{\mathrm{L}}(t) \exp\left(-\dnsty \mathrm{M}_\mathrm{N}(t)\right)
$.
Similarly, the PDF of propagation loss given the associated BS is NLOS is given by 
$
\tilde{f}_\mathrm{N}(t) = \frac{\mathrm{B}_\mathrm{N}}{\mathrm{A}_\mathrm{N}}f_{\mathrm{N}}(t) \exp\left(-\dnsty \mathrm{M}_\mathrm{L}(t)\right)
$. Define $\mathcal{S}(\tau,\mathrm{U}) \triangleq \pr\left(\SNR_{\x,0}>\tau|U_\x=\mathrm{U}\right)$. Thus, $\overline{\mathcal{S}}(\tau) = \cexpect{\mathcal{S}(\tau,\mathrm{U})}{U_\x=\mathrm{U}}$. 
By the law of total probability,
\ifSC
\begin{align*}
\mathcal{S}(\tau,\mathrm{U}) &= \mathrm{A}_\mathrm{L} \pr\left(\SNR_{x,0}>\tau\Big| \text{LOS connection} \right) +\mathrm{A}_\mathrm{N} \pr\left(\SNR_{x,0}>\tau\Big| \text{NLOS connection} \right)\\
& \stackrel{(a)}{\approx} \mathrm{A}_\mathrm{L} \int_{0}^{\infty} \pr\left(|\gamma_{i_m,\x,0}|^2>\frac{\eta_{\mathrm{L}}\tau \mathrm{U} t \sigma^2_n }{p_{\mathrm{MU}}\mathrm{G}}\right) \tilde{f}_\mathrm{L}(t) \mathrm{d} t + \mathrm{A}_\mathrm{N} \int_{0}^{\infty} \pr\left(|\gamma_{i_m,\x,0}|^2>\frac{\eta_{\mathrm{N}}\tau \mathrm{U} t \sigma^2_n }{p_{\mathrm{MU}}\mathrm{G}}\right) \tilde{f}_\mathrm{N}(t) \mathrm{d} t\\
&= \mathrm{B}_\mathrm{L}\int_{0}^{\infty} \pr\left(|\gamma_{i_m,\x,0}|^2>\frac{\eta_{\mathrm{L}}\tau \mathrm{U} t \sigma^2_n}{p_{\mathrm{MU}}\mathrm{G}}\right) \exp\left(-\dnsty \mathrm{M}_\mathrm{N}(t)\right)f_{\mathrm{L}}(t)\mathrm{d}t \\
&\hspace{3cm}+ \mathrm{B}_\mathrm{N}\int_{0}^{\infty} \pr\left(|\gamma_{i_m,x,0}|^2>\frac{\eta_{\mathrm{N}}\tau \mathrm{U} t \sigma^2_n}{p_{\mathrm{MU}}\mathrm{G}}\right) \exp\left(-\dnsty \mathrm{M}_\mathrm{L}(t)\right)f_{\mathrm{N}}(t)\mathrm{d}t,
\end{align*} 
\else
\begin{align*}
&\mathcal{S}(\tau,\mathrm{U}) = \mathrm{A}_\mathrm{L} \pr\left(\SNR_{x,0}>\tau\Big| \text{LOS connection} \right) \\
&\hspace{1cm}+\mathrm{A}_\mathrm{N} \pr\left(\SNR_{x,0}>\tau\Big| \text{NLOS connection} \right)\\
& \stackrel{(a)}{\approx} \mathrm{A}_\mathrm{L} \int_{0}^{\infty} \pr\left(|\gamma_{i_m,\x,0}|^2>\frac{\eta_{\mathrm{L}}\tau \mathrm{U} t \sigma^2_n }{p_{\mathrm{MU}}\mathrm{G}}\right) \tilde{f}_\mathrm{L}(t) \mathrm{d} t \\
&\hspace{1cm}+ \mathrm{A}_\mathrm{N} \int_{0}^{\infty} \pr\left(|\gamma_{i_m,\x,0}|^2>\frac{\eta_{\mathrm{N}}\tau \mathrm{U} t \sigma^2_n }{p_{\mathrm{MU}}\mathrm{G}}\right) \tilde{f}_\mathrm{N}(t) \mathrm{d} t \\
&= \mathrm{B}_\mathrm{L}\int\limits_{0}^{\infty} \pr\left(|\gamma_{i_m,\x,0}|^2>\frac{\eta_{\mathrm{L}}\tau \mathrm{U} t \sigma^2_n}{p_{\mathrm{MU}}\mathrm{G}}\right) e^{-\dnsty \mathrm{M}_\mathrm{N}(t)}f_{\mathrm{L}}(t)\mathrm{d}t \\
&+ \mathrm{B}_\mathrm{N}\int\limits_{0}^{\infty} \pr\left(|\gamma_{i_m,x,0}|^2>\frac{\eta_{\mathrm{N}}\tau \mathrm{U} t \sigma^2_n}{p_{\mathrm{MU}}\mathrm{G}}\right) e^{-\dnsty \mathrm{M}_\mathrm{L}(t)}f_{\mathrm{N}}(t)\mathrm{d}t,
\end{align*} 
\fi
where (a) is obtained using Proposition~\ref{prop:1}. Note that the first integral is the probability that $\SNR$ exceeds the threshold and there is LOS connection, whereas the second term is for NLOS connection. Let us consider the probabilities in each of these two terms separately.
\begin{align*}
&\pr\left(|\gamma_{i_m,x,0}|^2>\frac{\eta_{\mathrm{L}}\tau \mathrm{U} t  \sigma^2_n}{p_{\mathrm{MU}}\mathrm{G}}\right) \iftoggle{SC}{}{\\&} = \pr\left(p_\mathrm{MU} = 1\right) \pr\left(|\gamma_{i_m,x,0}|^2>\frac{\eta_{\mathrm{L}}\tau \mathrm{U} t \sigma^2_n}{p_{\mathrm{MU}}\mathrm{G}}\Big| p_{\mathrm{MU}}=1\right)\\&
\stackrel{(b)}{=} \zeta(\eta_\mathrm{L},\mathrm{U}) \pr\left(|\gamma_{i_m,x,0}|^2>\frac{\eta_{\mathrm{L}}\tau \mathrm{U} t \sigma^2_n}{\mathrm{G}}\right),
\end{align*}
where (b) is obtained from distribution of $p_\mathrm{MU}$ in Proposition~\ref{prop:1}. 

Further, using the distribution of maximum of $\eta_\mathrm{L}$ exponential random variables for $|\gamma_{i_m,x,0}|^2$, 
\begin{align*}
\iftoggle{SC}{}{&}\pr\left(|\gamma_{i_m,x,0}|^2>\frac{\eta_{\mathrm{L}}\tau \mathrm{U} t \sigma^2_n}{\mathrm{G}}\right) \iftoggle{SC}{}{\\}
&= \sum_{n=1}^{\eta_\mathrm{L}}(-1)^{n+1}{\eta_\mathrm{L} \choose n}\exp\left(-\eta_{\mathrm{L}}\tau n \mathrm{U} t  \sigma^2_n/\mathrm{G}\right).
\end{align*}
Similarly, we can find the NLOS probability term, which completes the proof.

\subsection{Proof of Lemma~\ref{lem:int}}
\label{app:int}
The Laplace functional of the out-of-cell interference to a user at origin, given the path loss to the serving BS, is defined as 
$
L_{I_0,l}(s) \triangleq \expect{\exp\left(-s I_0\right)|L(\x,0) = l}.
$

\begin{align*}
L_{I_0,l}(s) &= \expect{\exp\left(-s\sum_{y\in\PPP, y\neq \x}\frac{\mathrm{G} |\gamma_{y,0}|^2 \chi_y }{L(y,0) U_y}\right)\Bigg|L(\x,0) = l}\\
 & \stackrel{(a)}{=} \expect{\exp\left(-s\sum_{t\in\mathcal{N}, t\geq l}\frac{\mathrm{G} |\gamma_{t}|^2 t^{-1} \chi_t }{U_t}\right)}
 \iftoggle{SC}{}{\\&}
  \stackrel{(b)}{=} \expect{\prod_{t\in\mathcal{N}, t\geq l}\exp\left(-\frac{s\mathrm{G} |\gamma_{t}|^2 t^{-1}\chi_t }{U_t}\right)}\\
%\end{align*}
%\begin{align*}
 & {=} \expect{\prod_{t\in\mathcal{N}, t\geq l}\cexpect{\exp\left(-\frac{s\mathrm{G} |\gamma_{t}|^2 t^{-1}\chi_t }{U_t}\right)}{|\gamma_{t}|^2}}
  \iftoggle{SC}{}{\\&}
 \stackrel{(c)}{=} \expect{\prod_{t\in\mathcal{N}, t\geq l}\frac{1}{1+\psi_t}}\\
& \stackrel{(d)}{=} \exp\left(-\int_{l}^{\infty}\left(1-\cexpect{\frac{1}{1+\psi_t}}{\psi_t}\right)\Lambda(dt)\right)
 \iftoggle{SC}{}{\\&}
 = \exp\left(-\int_{l}^{\infty}\left(\cexpect{\frac{1}{1+\psi_t^{-1}}}{\psi_t}\right)\Lambda(dt)\right),
\end{align*}
where (a) is obtained by displacing each point $y\in\PPP, y\neq \x$ to $L(y,0) = t \in \mathcal{N}$, $t\geq l$. Note that $\gamma_{y,0}, U_y$ and $\chi_y$ are independent marks of $y\in\PPP$, whose distributions are themselves independent of the location $y$. After one to one mapping of each point $y\in\mathbb{R}^2$ to $t\in\mathbb{R}^+$ and each mark to itself, we associate each feasible point $t\in\mathcal{N}$ with independent marks $\gamma_{t}, U_t$ and $\chi_t$, with same distribution as the corresponding earlier marks. Here, (b) is obtained using independence of the marks of the displaced PPP and (c) since $\gamma_t$ are exponentially distributed random variables with unit mean and $\psi_t = \frac{s\mathrm{G}t^{-1}\chi_t}{U_t}$. Using the PGFL (probability generating functional) \cite{BacBook09} we obtain (d). Using the distribution of $U_y$ and $\chi_t$, we get the required result.

\subsection{Laplace Functional of Out-of-cell Interference for General Number of Paths}
\label{app:int2}
The out-of-cell interference from a BS at $y$ to user at origin, served by BS at $\x$ is given by
\ifSC
\begin{align*}
I_{y,0}&= \frac{\mathrm{G}L(y,0)^{-1}}{\eta_{y,u}U_y} \sum_{w\in\mathcal{U}_y} ||\sum_{j=1}^{\eta_{y,u}}\gamma_{j} \aMS^*(\phi_{x,u})\aMS(\phi_{j,y,u})\aBS^*(\theta_{j,y,u})\aBS(\theta_{y,w})||^2\\
& = \frac{\mathrm{G}L(y,0)^{-1}}{\eta_{y,u}U_y} \sum_{w\in\mathcal{U}_y} ||\sum_{j=1}^{\eta_{y,u}}\gamma_{j}\chi_{j,w}||^2,
\end{align*}
\else
\begin{multline*}
I_{y,0}= \frac{\mathrm{G}L(y,0)^{-1}}{\eta_{y,u}U_y} \sum_{w\in\mathcal{U}_y} 
||\sum_{j=1}^{\eta_{y,u}}\gamma_{j} \aMS^*(\phi_{x,u})\aMS(\phi_{j,y,u})\times \\\hspace{2cm}\aBS^*(\theta_{j,y,u})\aBS(\theta_{y,w})||^2,
\end{multline*}
Thus,
\begin{align*}
I_{y,0} = \frac{\mathrm{G}L(y,0)^{-1}}{\eta_{y,u}U_y} \sum_{w\in\mathcal{U}_y} ||\sum_{j=1}^{\eta_{y,u}}\gamma_{j}\chi_{j,w}||^2,
\end{align*}
\fi
where $\chi_{j,w}$ is given by,
%\begin{equation*}
%\chi_{j,w} = \begin{cases}
%1 & \text{ if } \phi_{x,u} = \phi_{j,y,u} \text{ and } \theta_{y,w} = \theta_{j,y,u}\\
%\rho_{\mathrm{BS}} & \text{ if } \phi_{x,u} = \phi_{j,y,u} \text{ and } \theta_{y,w} \neq \theta_{j,y,u}\\
%\rho_{\mathrm{UE}} & \text{ if } \phi_{x,u} \neq \phi_{j,y,u} \text{ and } \theta_{y,w} = \theta_{j,y,u}\\
%\rho_{\mathrm{BS}} \rho_{\mathrm{UE}} & \text{ otherwise.} 
%\end{cases}
%\end{equation*}
\ifSC
%\begin{multline*}
%\chi_{j,w} = \mathds{1}_{\{\phi_{x,u} = \phi_{j,y,u}\}}\mathds{1}_{\{\theta_{y,w} = \theta_{j,y,u}\}} +
%\rho_{\mathrm{BS}}\mathds{1}_{\{\phi_{x,u} = \phi_{j,y,u}\}}\mathds{1}_{\{\theta_{y,w} \neq \theta_{j,y,u}\}} +\\
%\rho_{\mathrm{UE}} \mathds{1}_{\{\phi_{x,u} \neq \phi_{j,y,u}\}}\mathds{1}_{\{\theta_{y,w} = \theta_{j,y,u}\}} + 
%\rho_{\mathrm{BS}} \rho_{\mathrm{UE}} \mathds{1}_{\{\phi_{x,u} \neq \phi_{j,y,u}\}} \mathds{1}_{\{\theta_{y,w} \neq \theta_{j,y,u}\}}
%\end{multline*}
\begin{equation*}
\chi_{j,w} = \begin{cases}
1 & \text{ if } \phi_{x,u} = \phi_{j,y,u} \text{ and } \theta_{y,w} = \theta_{j,y,u}\\
\rho_{\mathrm{BS}} & \text{ if } \phi_{x,u} = \phi_{j,y,u} \text{ and } \theta_{y,w} \neq \theta_{j,y,u}\\
\rho_{\mathrm{UE}} & \text{ if } \phi_{x,u} \neq \phi_{j,y,u} \text{ and } \theta_{y,w} = \theta_{j,y,u}\\
\rho_{\mathrm{BS}} \rho_{\mathrm{UE}} & \text{ otherwise.} 
\end{cases}
\end{equation*}
\else
\begin{equation*}
\chi_{j,w} = \begin{cases}
1 & \text{ if } \phi_{x,u} = \phi_{j,y,u} \text{ and } \theta_{y,w} = \theta_{j,y,u}\\
\rho_{\mathrm{BS}} & \text{ if } \phi_{x,u} = \phi_{j,y,u} \text{ and } \theta_{y,w} \neq \theta_{j,y,u}\\
\rho_{\mathrm{UE}} & \text{ if } \phi_{x,u} \neq \phi_{j,y,u} \text{ and } \theta_{y,w} = \theta_{j,y,u}\\
\rho_{\mathrm{BS}} \rho_{\mathrm{UE}} & \text{ otherwise.} 
\end{cases}
\end{equation*}
\fi
Now let us look at the Laplace functional of this interference power. 
\begin{align*}
\iftoggle{SC}{}{&} L_{I_0}(s) \iftoggle{SC}{&}{}= \expect{\exp\left(-s \sum_{y\in \PPP, y\neq x} I_{y,0} \right)} \iftoggle{SC}{}{\\&}
 = \mathbb{E}\left[\exp\left(-s \sum_{y\in\PPP, y\neq x} \frac{\mathrm{G} L(y,0)^{-1}}{\eta_{y,0}U_y} \iftoggle{SC}{}{\right.\right.\\&\hspace{3.5cm}\times \left.\left.} \sum_{w\in\mathcal{U}_y}  ||\sum_{j=1}^{\eta_{y,0}}\gamma_j \chi_{j,w}||^2\right)\right]\\
%& = \expect{\prod_{y\in\PPP,y\neq x} \exp\left(\frac{-s \mathrm{G} L(y,0)^{-1}}{\eta_{y,0}U_y}\sum_{w\in\mathcal{U}_y}||\sum_{j=1}^{\eta_{y,0}}\gamma_j \chi_{j,w}||^2 \right)}\\
& \stackrel{(a)}{=} \mathbb{E}\left[\prod_{y\in\PPP,y\neq x} \mathbb{E}_{{\chi_{(.,.)},\gamma_{(.)}}}\left[\exp\left(\frac{-s \mathrm{G} L(y,0)^{-1}}{\eta_{y,0}U_y} \iftoggle{SC}{}{\right.\right.\right.\\&\hspace{3.5cm}\times \left.\left.\left.}
\sum_{w\in\mathcal{U}_y}||\sum_{j=1}^{\eta_{y,0}}\gamma_j \chi_{j,w}||^2 \right)\right]\right].
\end{align*}
where (a) follows since $\chi$ and $\gamma$ have distributions independent of location $y$.
Finding the exact distribution from this expression is intractable. The main bottleneck is that the small scale fading random variables $\gamma_j$, are together clubbed in a single norm expression and thus, although these random variables are assumed to be independent, the distribution of the norm squared for different users in $\mathcal{U}_y$ are correlated exponential random variables. We, thus, find upper and lower bounds in this work. 
\subsubsection{Upper Bound on the Laplace Functional}
In order to find an upper bound, we use the fact that $\chi_{j,w}\geq \rho_\mathrm{BS} \rho_\mathrm{UE}$. Thus, 
\begin{align*}
\iftoggle{SC}{}{\hspace{-0.4cm}}L_{I_0}(s) &\leq   \mathbb{E}\left[\prod_{y\in\PPP,y\neq x} \mathbb{E}_{\chi_{(.,.)},\gamma_{(.)}}\left[\exp\left(-\frac{s\mathrm{G}  \rho^2_\mathrm{BS}\rho^2_\mathrm{UE}}{L(y,0)\eta_{y,0}U_y}
\iftoggle{SC}{}{\right.\right.\right.\\&\hspace{3.5cm}\times \left.\left.\left.}
 \sum_{w\in\mathcal{U}_y}||\sum_{j=1}^{\eta_{y,0}}\gamma_j ||^2 \right)\right]\right]\\
& \stackrel{(a)}{=}  \expect{\prod_{y\in\PPP,y\neq x} \cexpect{\exp\left(\frac{-s\mathrm{G}  \rho^2_\mathrm{BS}\rho^2_\mathrm{UE} \Xi}{L(y,0) \eta_{y,0}} \right)}{\Xi}}
\iftoggle{SC}{}{\\&}
= \expect{\prod_{y\in\PPP,y\neq x}\frac{1}{1+ s\mathrm{G} L(y,0)^{-1} \rho^2_\mathrm{BS}\rho^2_\mathrm{UE}}}, 
\end{align*}
where $\Xi$ is an exponential random variable with mean $\eta_{y,0}$ in (a).
In order to find the $\SINR$ distribution, we are interested in Laplace functional conditioned on path loss to serving BS. Thus, conditioning on $L(x,0)=l$ and displacing the points in $\Phi$ to  $\mathcal{N}$, similar to Appendix~\ref{app:int} we get,
\begin{align*}
L_{I_0,l}(s) \leq \expect{\prod_{t\in\mathcal{N},t\geq l}\frac{1}{1+ s\mathrm{G} t^{-1} \rho^2_\mathrm{BS}\rho^2_\mathrm{UE}}} 
\iftoggle{SC}{}{\\}
= \exp\left(-\int_{l}^{\infty} {\frac{\Lambda(\mathrm{d}t)}{1+\frac{1}{ s\mathrm{G} t^{-1} \rho^2_\mathrm{BS}\rho^2_\mathrm{UE}}}} \right).
\end{align*}
\subsubsection{Lower Bound on the Laplace Functional}
One obvious lower bound can be obtained using $\chi_{j,w}=1$.  The Laplace functional in this case is the same as for the upper bound with $\rho^2_\mathrm{BS}\rho^2_\mathrm{UE}$ replaced by 1. However, with the narrow beamwidth for a large number of antennas, this approximation is clearly very pessimistic. We can get a tighter lower bound using the Cauchy-Schwarz inequality as follows.  
\begin{align*}
L_{I_0}(s) &\geq \mathbb{E}\left[\prod_{y\in\Phi,y\neq x} \mathbb{E}_{\chi_{(.,.)},\gamma_{(.)}}\left[\exp\left(\frac{-s \mathrm{G} L(y,0)^{-1}}{\eta_{y,0}U_y} 
\iftoggle{SC}{}{\right.\right.\right.\\&\hspace{1.5cm}\times \left.\left.\left.}
\left(\sum_{j=1}^{\eta_{y,0}}||\gamma_j||^2 \right) \sum_{w\in\mathcal{U}_y}\sum_{j=1}^{\eta_{y,0}}\chi_{j,w}^2 \right)\right]\right]\\%{L(y,0),U_y, \eta_{y,0}}\\
%\end{align*}
%\begin{align*}
&\hspace{-1cm}=  \mathbb{E}\left[\prod_{y\in\Phi,y\neq x} \mathbb{E}_{\chi_{(.,.)}}\left[\left(1+\frac{s \mathrm{G} L(y,0)^{-1}}{\eta_{y,0}U_y}
\iftoggle{SC}{}{\right.\right.\right.\\&\hspace{3cm}\times \left.\left.\left.} \sum_{w\in\mathcal{U}_y}\sum_{j=1}^{\eta_{y,0}}\chi_{j,w}^2\right)^{-\eta_{y,0}}\right]\right].%{L(y,0),U_y, \eta_{y,0}}
%&=  \cexpect{\prod_{y\in\Phi,y\neq x} \cexpect{\left(\frac{1}{1+\frac{s \mathrm{G} L %(y,0)^{-1}}{\eta_{y,0}U_y} %\sum_{w\in\mathcal{U}_y}\sum_{j=1}^{\eta_{y,0}}\chi_{j,w}^2}\right)^{\eta_{y,0}}}{\phi_{x, %0},\phi_{:,y,0},\theta_{:,y,0},\theta_{y,w_{(.)}}}}{L(y,0),U_y, \eta_{y,0}}.
\end{align*}
Simplifying the term 
\iftoggle{SC}{$\Psi_y = \cexpect{\left({1+\frac{s \mathrm{G} L(y,0)^{-1}}{\eta_{y,0}U_y} \sum_{w\in\mathcal{U}_y}\sum_{j=1}^{\eta_{y,0}}\chi_{j,w}^2}\right)^{-\eta_{y,0}}}{\eta_{y,0}},$}{
\begin{equation*}
\Psi_y = \cexpect{\left({1+\frac{s \mathrm{G} L(y,0)^{-1}}{\eta_{y,0}U_y} \sum_{w\in\mathcal{U}_y}\sum_{j=1}^{\eta_{y,0}}\chi_{j,w}^2}\right)^{-\eta_{y,0}}}{\eta_{y,0}},
\end{equation*}
}
we get 
%\begin{multline*}
%\Psi_y =  \sum_{i=1}^{\NBS}\sum_{m=0}^{\eta_{y,0}}{\eta_{y,0} \choose m} q^{m+1}_{\mathrm{UE},i} (1-q_{\mathrm{UE},i})^{\eta_{y,0}-m}\mathbb{E}_{\theta_{:,y,0},\theta_{y,w_{(.)}}}\left[\left(1+\frac{s \mathrm{G} L(y,0)^{-1}}{\eta_{y,0}U_y} \right.\right.\times\\
%\left.\left. \sum\limits_{j=1}^{\eta_{y,0}}\left(\sum\limits_{w\in\mathcal{U}_y}\left((1-\rho^2_\mathrm{BS})\mathds{1}(\theta_{j,y,0}=\theta_{y,w}) +\rho^2_\mathrm{BS}\right)\right)\left(\rho^2_\mathrm{UE}+\mathds{1}(j\leq m)(1-\rho^2_\mathrm{UE})\right)\right)^{-\eta_{y,0}}\right].
%\end{multline*}
%This can be further simplified as 
\begin{multline*}
\Psi_y =  \sum_{i=1}^{\NBS}\sum_{m=0}^{\eta_{y,0}}{\eta_{y,0} \choose m} q^{m+1}_{\mathrm{UE},i} (1-q_{\mathrm{UE},i})^{\eta_{y,0}-m}
\iftoggle{SC}{}{\\\times}
\sum_{k_1,\ldots,k_{U_y}=1}^{\NBS}\sum_{j_1,\ldots,j_{U_y}=0}^{\eta_{y,0}}\prod_{n=1}^{U_y}{\eta_{y,0}\choose j_n} q^{j_n+1}_{\mathrm{BS},k_n}(1-q_{\mathrm{BS},k_n})^{\eta_{y,0}-j_n} \\
\times\left(1+\frac{s \mathrm{G} L(y,0)^{-1}}{\eta_{y,0}U_y} \sum_{j=1}^{\eta_{y,0}}\left(\sum_{n=1}^{U_y}\left(\rho^2_\mathrm{BS}
\iftoggle{SC}{}{\right.\right. \right.\\\left.\left.\left.}
+
\mathds{1}(j\leq j_n))(1-\rho^2_\mathrm{BS}  \right)\right)\left(\rho^2_\mathrm{UE}+\mathds{1}(j\leq m)(1-\rho^2_\mathrm{UE})\right)\right)^{-\eta_{y,0}}.
\end{multline*}
The above expression boils down to Lemma~\ref{lem:int}, for a single path channel.  
This expression can be further simplified assuming equiprobable virtual angles,
\begin{multline*}
\Psi_y =  \sum_{m=0}^{\eta_{y,0}}{\eta_{y,0} \choose m} \left(\frac{1}{\NMS}\right)^{m} \left(1-\frac{1}{\NMS}\right)^{\eta_{y,0}-m}
\iftoggle{SC}{}{\\ \times}
 \sum_{j_1,\ldots,j_{U_y}=0}^{\eta_{y,0}}\prod_{n=1}^{U_y}{\eta_{y,0}\choose j_n} \left(\frac{1}{\NBS}\right)^{j_n}\left(1-\frac{1}{\NBS}\right)^{\eta_{y,0}-j_n} \\
\times\left(1+\frac{s \mathrm{G} L(y,0)^{-1}}{\eta_{y,0}U_y} \sum_{j=1}^{\eta_{y,0}}
\left(\sum_{n=1}^{U_y}\left(\rho^2_\mathrm{BS} +\mathds{1}(j\leq j_n)) 
\iftoggle{SC}{}{\right.\right.\right.\\\left.\left.\left.\times}
(1-\rho^2_\mathrm{BS}  \right)\right)\left(\rho^2_\mathrm{UE}
+\mathds{1}(j\leq m)(1-\rho^2_\mathrm{UE})\right)\right)^{-\eta_{y,0}}.
\end{multline*}
Now seperating the LOS and NLOS terms and using the Displacement theorem as for the upper bound, the Laplace functional can be given as 
\begin{align*}
L_{I_0,l}(s) \geq \exp\left(-\int_{l}^{\infty}(1-\expect{\Psi_{t,\mathrm{L}}})\Lambda_\mathrm{L}(\mathrm{d}t) \right)
\iftoggle{SC}{}{\\ \times}
\exp\left(-\int_{l}^{\infty}(1-\expect{\Psi_{t,\mathrm{N}}})\Lambda_\mathrm{N}(\mathrm{d}t) \right).
\end{align*}
where $\Psi_{t,\mathrm{j}}$ is same as $\Psi_{y}$ with $y$ replaced by $t$ and $\eta_{y,0}$ replaced by $\eta_\mathrm{j}$, for $j\in\{\mathrm{L},\mathrm{N}\}$. 
\end{appendix}
\section*{Acknowledgment}
The authors thank A. Alkhateeb for his comments on early drafts of this work. The authors also thank anonymous reviewers for their constructive comments. 
%\addcontentsline{toc}{section}{Acknowledgment}
\bibliographystyle{IEEEtran}
\bibliography{IEEEabrv,nextgen}
\end{document}